\documentclass[aps,prb,twocolumn,superscriptaddress,longbibliography]{revtex4-2}

\usepackage{amssymb,amsmath,amsthm,bbm,bbold}
\usepackage{amsfonts}
\usepackage{graphicx}
\usepackage{mathtools}
\usepackage[breaklinks,colorlinks,allcolors=blue]{hyperref}
\usepackage[normalem]{ulem}
\usepackage[T1]{fontenc}
\usepackage{color}
\usepackage{xcolor}
\usepackage{tcolorbox}

\usepackage{pifont}
\usepackage{amssymb}
\usepackage{tikz}
\usepackage{makecell}

\def\RR{\mathbbm{R}}

\newcommand{\bd}[1]{\boldsymbol{#1}}
\newcommand{\wb}{\bd{w}}

\usepackage{accents}

\newcommand{\Tr}{\mathrm{Tr}}

\newcommand*\xbar[1]{\hbox{\vbox{
       \hrule height 0.6pt 
       \kern0.3ex
       \hbox{%
         \kern-0.2em
         \ensuremath{#1}%
         \kern 0.0em
         }}}}

\newcommand*\xxbar[1]{\hbox{\vbox{
       \hrule height 0.6pt 
       \kern0.3ex
       \hbox{%
         \kern-0.0em
         \ensuremath{#1}%
         \kern 0.0em
         }}}}

\newcommand{\Ebw}{\,\xxbar{\mathcal{E}}^N\!(\wb)}
\newcommand{\ebw}{\,\xxbar{\mathcal{E}}^1_N(\wb)}

\newcommand{\bra}[1]{\mbox{$\langle #1 |$}}
\newcommand{\ket}[1]{\mbox{$| #1 \rangle$}}

\newtheorem{thm}{Theorem}
\newtheorem{lem}[thm]{Lemma}

\begin{document}
\title{Refining ensemble $N$-representability of one-body density matrices \\ from partial information}

\author{Julia Liebert}
\affiliation{Department of Physics, Arnold Sommerfeld Center for Theoretical Physics, Ludwig-Maximilians-Universität München, Theresienstrasse 37, 80333 München, Germany}
\affiliation{Munich Center for Quantum Science and Technology (MCQST), Schellingstrasse 4, 80799 München, Germany}
\author{Anna O. Schouten}
\affiliation{Department of Chemistry and The James Franck Institute, The University of Chicago, Chicago, IL 60637}
\author{Irma Avdic}
\affiliation{Department of Chemistry and The James Franck Institute, The University of Chicago, Chicago, IL 60637}
\author{Christian Schilling}
\email{c.schilling@physik.uni-muenchen.de}
\affiliation{Department of Physics, Arnold Sommerfeld Center for Theoretical Physics, Ludwig-Maximilians-Universität München, Theresienstrasse 37, 80333 München, Germany}
\affiliation{Munich Center for Quantum Science and Technology (MCQST), Schellingstrasse 4, 80799 München, Germany}
\author{David A. Mazziotti}
\email{damazz@uchicago.edu}
\affiliation{Department of Chemistry and The James Franck Institute, The University of Chicago, Chicago, IL 60637}
\date{Submitted June 11, 2025}

\begin{abstract}
The $N$-representability problem places fundamental constraints on reduced density matrices (RDMs) that originate from physical many-fermion quantum states. Motivated by recent developments in functional theories, we introduce a hierarchy of ensemble one-body $N$-representability problems that incorporate partial knowledge of the one-body reduced density matrices (1RDMs) within an ensemble of $N$-fermion states with fixed weights $w_i$. Specifically, we propose a systematic relaxation that reduces the refined problem---where full 1RDMs are fixed for certain ensemble elements---to a more tractable form involving only natural occupation number vectors. Remarkably, we show that this relaxed problem is related to a generalization of Horn’s problem, enabling an explicit solution by combining its constraints with those of the weighted ensemble $N$-representability conditions. An additional convex relaxation yields a convex polytope that provides physically meaningful restrictions on lattice site occupations in ensemble density functional theory for excited states.
\end{abstract}


\maketitle

\section{Introduction\label{sec:intro}}

The $N$-representability problem for reduced density matrices (RDMs) is a fundamental challenge in quantum chemistry and many-body quantum physics, particularly in electronic structure theory~\cite{C63, GP64, Kummer67, Erdahl78, Coleman-book, KL06, AK08, M12, M12-pra, Maciazek_2020, Schilling_2020, M23}. While computing an RDM from a known quantum state is straightforward, the inverse problem---determining whether a given RDM corresponds to an actual $N$-particle quantum state---is significantly more difficult. This difficulty is formalized by the QMA-completeness of the two-body $N$-representability problem~\cite{LCV07, OIWF22}. Even for the one-particle reduced density matrix (1RDM), compatibility with a pure $N$-fermion quantum state imposes spectral constraints that are difficult to compute and scale poorly with system size~\cite{Klyachko2004, KL06, AK08, SGC13, S15, CS17stability, SAKLWCR18, Maciazek_2020, Schilling_2020}. Nonetheless, $N$-representability constraints play a crucial role in both method development in quantum chemistry~\cite{C63, M04, Mazziotti.2011mwm, Schilling18, P21, SSM22, SSM22-2, JLDV23, DP24} and in quantum algorithms for correlated systems~\cite{RBM18, SM20, Google20, RM20, AM24, AM24-pra}.

The one-body $N$-representability problem is especially important in the context of functional theories, where it defines the domain of the universal functional in ground-state reduced density matrix functional theory (RDMFT)~\cite{GPB05, PG16, KSPB16, SKB17, SB18, BM18, Piris19, MPU19, SBM19, Piris21-1, P21, DVR21, LKO22, GBM22, MP22, SNF22, LCS23, SG23, GBM23, GBM24, BWR24, VMSS24, CS24, CG24, YS24}. A 1RDM is compatible with an ensemble of $N$-fermion quantum states if and only if it satisfies the Pauli exclusion principle and is normalized to the particle number~\cite{Watanabe39, C63}. To extend functional theories to excited states, one generalizes the Rayleigh-Ritz variational principle to mixed states of the form $\Gamma = \sum_i w_i \ket{\Psi_i}\!\bra{\Psi_i}$, where the weights \( \boldsymbol{w} \) are fixed~\cite{GOK88a, LHS24}. In this setting, the domain of the universal functional in $\boldsymbol{w}$-RDMFT~\cite{SP21, LCLS21, LS23-njp, LS23-sp} and lattice ensemble density functional theory (EDFT)~\cite{GOK88c, YPBU17, Fromager2020-DD, Loos2020-EDFA, Cernatic21, GK21, Yang21, GKGGP23, GL23, SKCPJB24, CPSF24} is determined by the solution to the $\boldsymbol{w}$-ensemble one-body $N$-representability problem~\cite{SP21, LCLS21, CLLPPS23, LCLMS24}, whose refinement through additional physical information is the focus of this work.

In this paper, we show that \textit{partial information} about the ground and/or excited states---specifically their one-body reduced density matrices (1RDMs)---can be used to significantly refine the set of admissible 1RDMs in ensemble calculations for excited states. Building on the traditional $\boldsymbol{w}$-ensemble framework, in which an ensemble of quantum states contributes with fixed weights $w_i$, we formulate a refined ensemble one-body $N$-representability problem that incorporates not only the standard representability constraints but also additional structure imposed by fixing one or more of the 1RDMs associated with selected contributing states, such as the ground state $|\Psi_1\rangle$ and low-lying excited states $|\Psi_i\rangle$.

This formulation generalizes Coleman's classical solution, which determines the conditions under which a given 1RDM is representable by an ensemble of $N$-fermion states, without assuming any knowledge of the contributing states. In contrast, the refined problem considered here assumes not only that the contributing states have fixed weights $w_i$, but also that the 1RDMs of some of these states are known. This partial information leads to a strictly smaller and more physically meaningful set of ensemble 1RDMs by imposing more stringent $N$-representability constraints on the unknown remainder of the ensemble. The resulting framework is particularly relevant in practical excited-state simulations, where contributing states are often computed sequentially and partial one-body information is readily available.


The refined $\boldsymbol{w}$-ensemble one-body $N$-representability problem is highly complex, largely because its solution depends on the natural orbitals of the fixed 1RDMs. As a result, obtaining a complete solution is generally impractical. To address this challenge, we introduce a systematic relaxation in which the fixed 1RDMs are replaced by their natural occupation number vectors, discarding information about the associated orbitals. The resulting relaxed set is characterized purely by spectral data and can be fully described by combining the solution to a generalized version of Horn’s problem~\cite{Knutson, Bhatia, LP03} with the $\boldsymbol{w}$-ensemble $N$-representability constraints~\cite{SP21, LCLS21, CLLPPS23, LCLMS24}. While this yields a complete solution to the relaxed problem, the number of spectral constraints grows rapidly with system size, limiting its scalability.

To overcome this limitation, we derive an outer approximation that provides a set of system-size independent constraints. These constraints remain effective regardless of the number of particles or the dimensionality of the one-particle Hilbert space, making them particularly valuable for quantum chemical applications where large basis sets are required. The tightness of this approximation depends on the deviation of the natural occupation number vectors from the Hartree--Fock point, and thus reflects the degree of correlation in the corresponding $N$-fermion quantum states.

This paper is organized as follows. Section~\ref{sec:recap} introduces the notation and terminology for $\boldsymbol{w}$-ensembles that will be used throughout the paper. Section~\ref{sec:Nrep} defines the refined $\boldsymbol{w}$-ensemble one-body $N$-representability problem and outlines the relaxation strategy that yields a practical characterization of the admissible 1RDMs. In Section~\ref{sec:residuals}, we derive the system-size independent constraints and examine their implications in molecular systems, focusing on ensembles with weakly and strongly correlated ground states. Finally, Section~\ref{sec:spectral} shows how taking the convex hull of the spectral constraints leads to a convex set that also governs the lattice site occupation number vectors in ensemble DFT for excited states.

\section{Short recap of $\bd w$-ensembles \label{sec:recap}}

In this section, we briefly introduce the notation and key aspects of the $\bd w$-ensemble one-body $N$-representability problem \cite{SP21, LCLS21, CLLPPS23} as used in this work. We then reformulate the problem in terms of probabilistic quantum channels, placing it within the broader context of quantum marginal problems studied in the quantum information community \cite{Klyachko2004, KL06, LCV07, ETRS08, WDGC13, Schilling2014, CDKW14, SOK14, WHG17, BCMW17, CS17stability, MV17, CSW18, YSWNG21, Haapasalo2021, HLA22, GHKPU23}.

We consider systems of $N$ non-relativistic fermions. The $N$-fermion Hilbert space, $\mathcal{H}_N$, is defined as the $N$-fold wedge product of the one-particle Hilbert space $\mathcal{H}_1$, i.e., $\mathcal{H}_N = \wedge^N\mathcal{H}_1$. Throughout this work, we assume $\mathcal{H}_1$ is finite-dimensional. The set of pure quantum states on $\mathcal{H}_N$ is denoted by $\mathcal{P}^N$, and the set of ensemble quantum states, $\mathcal{E}^N$, is the convex hull of the pure states, i.e., $\mathcal{E}^N = \mathrm{conv}(\mathcal{P}^N)$.

The one-particle reduced density operator (1RDM) of an $N$-fermion quantum state $\Gamma \in \mathcal{P}^N, \mathcal{E}^N$ is defined as
\begin{equation}
\gamma \equiv N\mathrm{Tr}_{N-1}[\Gamma]\,.
\end{equation}
Due to the linearity of the partial trace, the set $\mathcal{E}^1_N$ of 1RDMs compatible with ensemble $N$-fermion states is given by
\begin{equation}
\mathcal{E}^1_N\equiv N\Tr_{N-1}[\mathcal{E}^N]\,,
\end{equation}
which is also convex. It is the convex hull of the set $\mathcal{P}^1_N$ of 1RDMs compatible with pure $N$-fermion states, i.e.,
\begin{equation}\label{eq:E1N}
\mathcal{E}^1_N = \mathrm{conv}(\mathcal{P}^1_N)\,.
\end{equation}
The 1RDMs in $\mathcal{P}^1_N$ are called pure state one-body $N$-representable, while the 1RDMs in $\mathcal{E}^1_N$ are referred to as ensemble one-body $N$-representable. Due to the convexity property in Eq.~\eqref{eq:E1N}, the ensemble one-body $N$-representability problem is also known as the convex (relaxed) one-body $N$-representability problem, a terminology that we adopt in Sec.~\ref{sec:Nrep}.

According to Watanabe and Coleman \cite{Watanabe39, C63}, the set $\mathcal{E}^1_N$ of ensemble one-body $N$-representable 1RDMs is fully characterized by the Pauli exclusion principle and the normalization $\Tr_1[\gamma] = N$. In contrast, the set $\mathcal{P}^1_N$ of pure state one-body $N$-representable 1RDMs is determined by the generalized Pauli constraints \cite{KL06, AK08, K09, SGC13, S15, SAKLWCR18, Reuvers21}.

Clearly, pure states are quantum states with a fixed spectrum $\bd w_0 = (1, 0, \dots)$. The set of $N$-fermion quantum states with a generic fixed spectrum $\bd w$ is denoted by $\mathcal{E}^N(\bd w)$, where $\bd w \in \mathbbm{R}^D$ with $1\geq w_1 \geq w_2 \geq \dots \geq w_{D}\geq0$ is a decreasingly ordered vector with $\sum_i w_i = 1$, and $D = \mathrm{dim}(\mathcal{H}_N)$. In general, the set $\mathcal{E}^N(\bd w)$ is not convex, as the sum of two states with fixed spectra may result in a state with a different spectrum. Therefore, we also introduce its convex hull \cite{SP21, LCLS21, CLLPPS23},
\begin{equation}\label{eq:Ebw}
\Ebw \equiv \mathrm{conv}\left(\mathcal{E}^N(\bd w)\right)\,.
\end{equation}
Thus, any $\Gamma\in \Ebw$ can be written as a convex combination $\Gamma = \sum_kp_k\Gamma_k$ with $\Gamma_k=\sum_iw_i\ket{\Psi_i^{(k)}}\!\bra{\Psi_i^{(k)}}\in \mathcal{E}^N(\bd w)$ and $p_k\geq 0, \sum_k p_k=1$.
Due to this convex decomposition, we can further interpret any state $\Gamma\in \Ebw$ as the final state after the application of a probabilistic unitary quantum channel $\chi(\cdot)$ onto a fixed initial state $\Gamma_{\bd w}\in \mathcal{E}^N(\bd w)$ as (recall Uhlmann's theorem \cite{AU82})
\begin{equation}\label{eq:qc-Gw}
\chi(\Gamma_{\bd w}) = \sum_kp_k U_k\Gamma_{\bd w} U_k^\dagger\,,
\end{equation}
where $U_k$ are unitary operators on $\mathcal{H}_N$. Therefore, $\Ebw$ in Eq.~\eqref{eq:Ebw} can be interpreted as the set of states $\Gamma$ that result from applying probabilistic unitary quantum channels to a fixed but arbitrary $\Gamma_{\bd w} \in \mathcal{E}^N(\bd w)$.
In particular, Eq.~\eqref{eq:qc-Gw} provides a quantum information perspective on the concept of convex relaxation of non-convex sets as widely exploited in functional theories such as RDMFT \cite{Valone80, Lieb83, LCLS21, LS23-njp}.

The set of admissible 1RDMs compatible with an $N$-fermion state in either $\mathcal{E}^N(\bd w)$ or $\Ebw$ is defined as \cite{SP21, LCLS21}
\begin{equation}\label{eq:e1Nw}
\mathcal{E}^1_N(\bd w) = N\Tr_{N-1}[\mathcal{E}^N(\bd w)]\,,\,\,\ebw = N\Tr_{N-1}[\Ebw]\,.
\end{equation}
Furthermore, both sets $\mathcal{E}^1_N(\bd w)$ and $\ebw$ of 1RDMs are fully characterized by spectral constraints due to their invariance under one-particle unitary transformations. Specifically, for $\ebw$, the corresponding spectral set
\begin{equation}\label{eq:Sigma-w}
\Sigma(\bd w) \equiv \left\{ \bd\lambda\,|\,\exists \gamma\in \ebw, \pi\in\mathcal{S}_d: \bd\lambda =\pi( \mathrm{spec}^\downarrow(\gamma) ) \right\}
\end{equation}
is indeed a convex polytope as shown in Refs.~\cite{SP21, LCLS21, CLLPPS23}. Here, $\mathcal{S}_d$ denotes the symmetric group of degree $d=\mathrm{dim}(\mathcal{H}_1)$, and the map $\mathrm{spec}^\downarrow(\cdot)$ maps a linear operator on a finite-dimensional Hilbert space to its vector of decreasingly ordered eigenvalues. The hyperplane representation of $\Sigma(\bd w)$ led to the $\bd w$-ensemble one-body $N$-representability constraints, which are more restrictive than the Pauli exclusion principle \cite{SP21, LCLS21, CLLPPS23}.

\section{Refined $\bd w$-ensemble one-body N-representability problem \label{sec:Nrep}}

In this section, we formally introduce and define the refined $\bd w$-ensemble one-body $N$-representability problem, exploring its variants in terms of solvability and convexity. This will eventually lead to a purely spectral refined $\bd w$-ensemble one-body N-representability problem which is, however, less stringent than the original problem where the knowledge of the full 1RDMs of states in the $\bd w$-ensemble is assumed. While this section is more mathematical and technically detailed, it provides valuable insights that may enrich the understanding of the constraints; readers primarily interested in the constraints themselves may proceed directly to Sec.~\ref{sec:residuals}.

\subsection{Scientific problem\label{sec:origin}}

We consider an $N$-fermion state of the form $\Gamma = \sum_{i=1}^D w_i \ket{\Psi_i} \!\bra{\Psi_i}$, which represents a $\bd w$-ensemble $\{w_i, \ket{\Psi_i}\}$ consisting of fixed probabilities $w_i$ and corresponding pure states $\ket{\Psi_i}$. At this stage, we assume the states $\ket{\Psi_i}$ to be mutually orthogonal, i.e., $\langle \Psi_i | \Psi_j \rangle = 0$ for all $i \neq j$, motivated by the fact that they may represent eigenstates of a physical Hamiltonian. This orthogonality assumption will be relaxed later in our analysis.

In the following, we assume additional partial information about the ensemble beyond the fixed probabilities (also called weights) $w_i$, namely that the 1RDMs of some states $\ket{\Psi_i}$ are known as well.
As a simple example, suppose the 1RDM $\gamma^{(1)}$ of a non-degenerate eigenstate $\ket{\Psi_1}$ is known a priori.
Within the framework of the so-called GOK variational principle introduced by Gross, Oliveira and Kohn \cite{GOK88a} (see also Refs.~\cite{Theophilou79, LHS24}), $\gamma^{(1)}$ corresponds to the ground state $\ket{\Psi_1}$ of $H$. This situation arises, for instance, when $\gamma^{(1)}$ is obtained from a prior ground state calculation, and the GOK variational principle is then used to determine additional low-lying eigenenergies of $H$. The constraint imposed by $\gamma^{(1)}$ introduces additional restrictions on the sets $\mathcal{E}^N(\bd w)$ and $\mathcal{E}^1_N(\bd w)$, which we investigate in this paper.

For the general case of a generic number of known 1RDMs $\gamma^{(i)}$ of states $\ket{\Psi_i}$ in the $\bd w$-ensemble $\{w_i, \ket{\Psi_i}\}$, we introduce the index set $K$ as the set of increasingly ordered indices $i$ that label the fixed 1RDMs $\gamma^{(i)}$. By definition, the set $K$ has cardinality $|K| \leq r$, with $i \in \mathbbm{N}\backslash {0}$ and $1 \leq i \leq r$ for all $i \in K$. To keep track of these fixed 1RDMs, we define the tuple
\begin{equation}\label{eq:setKk}
\mathcal{K}_K\equiv \left(\gamma^{(i)}\right)_{i\in K}\,,\quad |K|\leq r\,,
\end{equation}
which contains the 1RDMs $\gamma^{(i)}$ that have been fixed.
To assign each index $i\in K$ with the corresponding fixed 1RDM $\gamma^{(i)}$, we introduce (on a technical level) the map
\begin{equation}
    \mu: K\to \{1, 2, ..., |K|\}\,.
\end{equation}
In particular, this will allow us in the following to write $\gamma^{(i)} = (\mathcal{K}_K)_{\mu(i)}$.

The knowledge of $\mathcal{K}_K$ in Eq.~\eqref{eq:setKk} then implies that the (non-convex) set $\mathcal{E}^N(\bd w)$ is restricted to the subset
\begin{equation}\label{eq:E1NKk}
\mathcal{E}^N(\bd w, \mathcal{K}_K)\equiv \left\{\Gamma\in \mathcal{E}^N(\bd w)| \forall i\in K\!: \ket{\Psi_i}\mapsto \gamma^{(i)} \right\}\,,
\end{equation}
where the map $\ket{\Psi_i}\mapsto \gamma^{(i)}$ is understood as the partial trace map $\gamma^{(i)}=N\Tr_{N-1}[\ket{\Psi_i}\!\bra{\Psi_i}]$.

The corresponding set of 1RDMs $\gamma$ that are compatible with an $N$-fermion state $\Gamma\in \mathcal{E}^N(\bd w, \mathcal{K}_K)$,
\begin{equation}\label{eq:e1NKk}
\mathcal{E}^1_N(\bd w, \mathcal{K}_K)\equiv N\Tr_{N-1}\left[ \mathcal{E}^N(\bd w,\mathcal{K}_K)\right]\,,
\end{equation}
is neither convex nor is it invariant under unitary transformations on the one-particle Hilbert space $\mathcal{H}_1$.
Therefore, the set $\mathcal{E}^1_N(\bd w, \mathcal{K}_K)$ is not solely characterized by spectral constraints, unlike $\mathcal{E}^1_N(\bd w)$ (see Eq.~\eqref{eq:e1Nw}), and its natural orbital dependence is generally difficult to determine. In fact, $\mathcal{E}^1_N(\bd w)$ includes the set $\mathcal{P}^1_N$ of pure state one-body $N$-representable 1RDMs as a special case, specifically for $\bd w_0 = (1, 0, \dots)$. Therefore, characterizing $\mathcal{E}^1_N(\bd w)$, and thus $\mathcal{E}^1_N(\bd w, \mathcal{K}_K) \subseteq \mathcal{E}^1_N(\bd w)$, is expected to be more complex than calculating the generalized Pauli constraints \cite{KL06, AK08, K09}. As a result, directly characterizing $\mathcal{E}^1_N(\bd w, \mathcal{K}_K)$ using refined $\bd w$-ensemble $N$-representability constraints that account for the knowledge of the 1RDMs in $\mathcal{K}_K$ proves to be highly challenging. To address this difficulty, it becomes necessary to systematically relax the definition of $\mathcal{E}^1_N(\bd w, \mathcal{K}_K)$ in Eq.~\eqref{eq:e1NKk}. This relaxation, which we will discuss in detail in the following, gives rise to an outer approximation of the original set.

According to the duality principle for compact convex sets (e.g., see \cite{R97}), such sets can be fully characterized by the intersection of their supporting hyperplanes. Moreover, replacing a set with its convex hull does not alter the result of minimizing or maximizing linear functionals over it, as is commonly done when solving the ground state problem in physics and quantum chemistry including RDMFT, except that the optimizer may no longer be unique. Thus, we first introduce the convex hull of $\mathcal{E}^N(\bd w, \mathcal{K}_K)$ in \eqref{eq:E1NKk} at the $N$-fermion level,
\begin{equation}\label{eq:Ebwg1}
\xxbar{\mathcal{E}}^N(\bd w, \mathcal{K}_k)\equiv \mathrm{conv}\left(\mathcal{E}^N(\bd w,  \mathcal{K}_k)\right)\,.
\end{equation}
To get a deeper understanding of the implications of taking the convex hull in Eq.~\eqref{eq:Ebwg1}, we note that any $\Gamma \in \xxbar{\mathcal{E}}^N(\bd w,\mathcal{K}_K)$ can be written as a convex combination
\begin{equation}\label{eq:qc}
\Gamma = \sum_{j} p_j \Gamma_j = \sum_{i=1}^rw_i\sum_{j}p_j\ket{\Psi_i^{(j)}}\!\bra{\Psi_i^{(j)}}\,,
\end{equation}
where $p_j\geq 0$, $\sum_j p_j=1$, and $\Gamma_j\in \mathcal{E}^N(\bd w,\mathcal{K}_K)$ for all $j$. Thus, the 1RDM of $\Gamma$ in Eq.~\eqref{eq:qc} follows from the linearity of the partial trace as
\begin{equation}
    \gamma  = \sum_{i=1}^rw_i \sum_{j}p_j \gamma^{(i)}_j
\end{equation}
with $\gamma^{(i)}_j \equiv N\Tr_{N-1}[\ket{\Psi_i^{(j)}}\!\bra{\Psi_i^{(j)}}]$. For the indices $i\in K$ (equivalently, $\gamma^{(i)}=(\mathcal{K}_K)_{\mu(i)}$), $\gamma^{(i)}_j = \gamma^{(i)}_{j^\prime}$ forall $j, j^\prime$ such that the 1RDMs associated with $w_i, i\in K$ in the summation in Eq.~\eqref{eq:qc} remain unchanged. For completeness, we note that this property of the set in Eq.~\eqref{eq:Ebwg1} with respect to the 1RDMs $\gamma^{(i)}=(\mathcal{K}_K)_{\mu(i)}$ can also be expressed using the probabilistic quantum channel perspective based on Eq.~\eqref{eq:qc-Gw}.

At the one-particle level, the set $\xxbar{\mathcal{E}}^1_N(\bd w, \mathcal{K}_K)$ of 1RDMs compatible with a state $\Gamma\in \xxbar{\mathcal{E}}^N(\bd w,\mathcal{K}_K)$ is given by
\begin{equation}
\xxbar{\mathcal{E}}^1_N(\bd w, \mathcal{K}_K)= \mathrm{conv}\left(\mathcal{E}^1_N(\bd w, \mathcal{K}_K)\right)\,.
\end{equation}
Since the 1RDMs $\gamma^{(i)} = (\mathcal{K}_K)_{\mu(i)}$ are fixed, the set $\xxbar{\mathcal{E}}^1_N(\bd w, \mathcal{K}_K)$ remains non-unitarily invariant and cannot be fully characterized by spectral constraints. In particular, its dependence on the natural orbitals becomes highly complex when multiple 1RDMs $\gamma^{(i)}$ with different natural orbital bases are fixed.

To make the problem more tractable, it is necessary to relax the orthonormality condition $\langle\Psi_i^{(j)}|\Psi_k^{(j)}\rangle = \delta_{ik}$, which must hold for all $i \in \{1, \ldots, r\}$ in Eq.~\eqref{eq:qc}, for the states $\ket{\Psi_i^{(j)}}$ that correspond to fixed 1RDMs $\gamma^{(i)} \in \mathcal{K}_K$. On the $N$-fermion level, this leads to the set
\begin{align}\label{eq:ENwKc}
    \widetilde{\mathcal{E}}^N(\bd w, \mathcal{K}_K)& \equiv \Big\{\Gamma\in\mathcal{E}^N\,\Big\vert\, \Gamma=\sum_{i\in K}w_i\Gamma^{(i)}+\widetilde{\Gamma},\nonumber\\
   &\quad\quad\Gamma^{(i)} \mapsto (\mathcal{K}_K)_{\mu(i)} ,\widetilde{\Gamma}\in\xxbar{\mathcal{E}}^N(\bd w_{K^c})\Big\}\,,
\end{align}
where we defined the complement $K^c = J\backslash K$ of the index set $K$ on the index set $J = \{1, 2, ..., r\}$, $\bd w_{K^c}\in \RR^{D-|K|}$ is the decreasingly ordered vector of weights $w_i$ that correspond to non-fixed $\gamma^{(i)}$, and
\begin{align}\label{eq:EbNwKc}
    \xxbar{\mathcal{E}}^N(\bd w_{K^c})\equiv \Big\{\Gamma &= \sum_{j}p_j\sum_{i\in K^c}w_i\ket{\Psi_i^{(j)}}\!\bra{\Psi_i^{(j)}}\,\Big\vert\, \nonumber\\
    &\quad \forall i,k\in K^c: \langle\Psi_i^{(j)}|\Psi_{k}^{(j)}\rangle=\delta_{ik} \Big\}\,.
\end{align}
Furthermore, the set of 1RDMs compatible with an $N$-fermion density operator $\Gamma\in \widetilde{\mathcal{E}}^N(\bd w, \mathcal{K}_K)$ is given by
\begin{equation}\label{eq:E1NwKc}
    \widetilde{\mathcal{E}}^1_N(\bd w, \mathcal{K}_K) \equiv N\Tr_{N-1}\left[\widetilde{\mathcal{E}}^N(\bd w, \mathcal{K}_K)\right]\,.
\end{equation}
Since $\gamma^{(i)} = (\mathcal{K}_K)_{\mu (i)}$, we can equivalently characterize the above set as
\begin{equation}\label{eq:tilde-e1NK}
    \widetilde{\mathcal{E}}^1_N(\bd w, \mathcal{K}_K) = \Big\{\gamma = \sum_{i\in K}w_i\gamma^{(i)}+\widetilde{\gamma}\,\Big\vert\, \widetilde{\gamma}\in \xxbar{\mathcal{E}}^1_N(\bd w_{K^c})\Big\}\,,
\end{equation}
where
\begin{equation}
    \xxbar{\mathcal{E}}^1_N(\bd w_{K^c})\equiv N\Tr_{N-1}\left[\xxbar{\mathcal{E}}^N(\bd w_{K^c})\right]
\end{equation}
is the set of 1RDMs compatible with a linear operator in $\xxbar{\mathcal{E}}^N(\bd w_{K^c})$.

By construction, the set $\widetilde{\mathcal{E}}^1_N(\bd w, \mathcal{K}_K)$ \eqref{eq:tilde-e1NK} satisfies the inclusion relation
\begin{equation}\label{eq:incl-E1sets}
\xxbar{\mathcal{E}}^1_N(\bd w,\mathcal{K}_K)\subseteq \widetilde{\mathcal{E}}^1_N(\bd w,\mathcal{K}_K)\,,
\end{equation}
as we merely relaxed the condition on 1RDMs $\gamma$ to be in the set $\xxbar{\mathcal{E}}^1_N(\bd w,\mathcal{K}_K)$ compared to $\widetilde{\mathcal{E}}^1_N(\bd w,\mathcal{K}_K)$. Furthermore, the set $\widetilde{\mathcal{E}}^1_N(\bd w,\mathcal{K}_K)$ is convex, as it follows from the convexity of $\widetilde{\mathcal{E}}^N(\bd w,\mathcal{K}_K)$, which can be readily proven, and the linearity of the partial trace map $N\Tr_{N-1}[\cdot]$.

To illustrate why the set $\widetilde{\mathcal{E}}^1_N(\bd w, \mathcal{K}_K)$ as an outer approximation to $\xxbar{\mathcal{E}}^1_N(\bd w,\mathcal{K}_K)$ is easier to characterize than the original set $\xxbar{\mathcal{E}}^1_N(\bd w, \mathcal{K}_K)$, we consider the case of $r = 2$ with $|K| = 1$, and we fix $\gamma^{(1)}$. The analysis for $\gamma^{(2)}$ fixed follows analogously.
Then, given a 1RDM $\gamma$, we have $\gamma\in \widetilde{\mathcal{E}}^1_N(\bd w, \gamma^{(1)})$ if and only if $\gamma\in \ebw$ and
\begin{equation}
\gamma^{(2)} = \frac{1}{1-w}\left(\gamma-w\gamma^{(1)}\right)\in \mathcal{E}^1_N\,.
\end{equation}
We observe that $\gamma^{(2)}$ is not well-defined, as $w \to 1$, consistent with $\bd w_0 = (1, 0, \ldots)$ representing an ensemble containing only ground state information and no excited state contributions. Both memberships $\gamma \in \ebw$ and $\gamma^{(2)} \in \mathcal{E}^1_N$ can be verified solely by evaluating spectral constraints, as detailed in Refs.~\cite{SP21, LCLS21, CLLPPS23}.
Thus,  for $r=2$, we obtain a practical approximation of the more intricate set $\xxbar{\mathcal{E}}^1_N(\bd w, \gamma^{(1)})$.

For general $r \geq 3$, however, determining a feasible characterization of the set $\widetilde{\mathcal{E}}^1_N(\bd w,\mathcal{K}_K)$ remains challenging. To obtain a set of 1RDMs defined solely through spectral constraints, we must discard information about the natural orbitals of the fixed 1RDMs in $\mathcal{K}_K$, retaining only their natural occupation number vectors $\bd\lambda^{(i)} = \pi(\mathrm{spec}^\downarrow(\gamma^{(i)})), \pi\in\mathcal{S}_d$, where $d=\mathrm{dim}(\mathcal{H}_1)$.

\subsection{Discarding information about natural orbitals\label{sec:natorb}}

In this section, we show that discarding information about the natural orbital bases of $\gamma^{(i)} =(\mathcal{K}_K)_{\mu(i)}$ allows for a complete characterization of the accordingly relaxed set $\widetilde{\mathcal{E}}^1_N(\bd w, \mathcal{K}_K)$ \eqref{eq:tilde-e1NK} by combining the solution to Horn's problem \cite{LP03, Knutson} with the $\bd w$-ensemble one-body $N$-representability constraints \cite{SP21, LCLS21, CLLPPS23}.

Neglecting the information about the natural orbitals of the 1RDMs in $\mathcal{K}_K$ and retaining only their spectral information implies considering the set
\begin{equation}\label{eq:setNONKk}
\mathcal{L}_K \equiv \Big(\bd\lambda^{(i)}\Big)_{i\in K}
\end{equation}
of fixed natural occupation number vectors $\bd\lambda^{(i)}$, similarly to $\mathcal{K}_K$ in Eq.~\eqref{eq:setKk}.
Using $\mathcal{L}_K$ \eqref{eq:setNONKk}, the set of admissible 1RDMs under this constraint is given by
\begin{equation}\label{eq:E1NL}
\widetilde{\mathcal{E}}^{1}_N(\bd w, \mathcal{L}_K)\!\equiv \!\bigcup_{u_1, ..., u_{|K|}\in \mathrm{U}(\mathcal{H}_1)} \widetilde{\mathcal{E}}^1_N\left(\bd w, \Big(u_iD(\bd\lambda^{(i)})u_i^\dagger\Big)_{i\in K}\right),
\end{equation}
where $\mathrm{U}(\mathcal{H}_1)$ is the group of unitary operators on $\mathcal{H}_1$, and $D(\bd\lambda^{(i)})$ denotes the diagonal matrix with matrix elements $D_{ij} = \delta_{ij}\lambda^{(i)}_j$. In particular, $\widetilde{\mathcal{E}}^{1}_N(\bd w, \mathcal{L}_K)$ satisfies the inclusion relation (recall Eq.~\eqref{eq:incl-E1sets})
\begin{equation}
   \xxbar{\mathcal{E}}^1_N(\bd w,\mathcal{K}_K)\subseteq \widetilde{\mathcal{E}}^1_N(\bd w,\mathcal{K}_K)\subseteq \widetilde{\mathcal{E}}^{1}_N(\bd w, \mathcal{L}_K)\,.
\end{equation}

The set $\widetilde{\mathcal{E}}^{1}_N(\bd w, \mathcal{L}_K)$ is eventually invariant under unitary transformations on the one-particle Hilbert space, meaning that for all 1RDMs $\gamma$ and one-particle unitaries $u:\mathcal{H}_1\to\mathcal{H}_1$, $\gamma\in\widetilde{\mathcal{E}}^{1}_N(\bd w, \mathcal{L}_K)$ implies that $u\gamma u^\dagger\in \widetilde{\mathcal{E}}^{1}_N(\bd w, \mathcal{L}_K)$. This follows directly from the fact that unitary transformations at the one-particle level preserve the spectrum of a 1RDM and that $\mathcal{E}^1_N$ is also unitarily invariant.
Consequently, the set $\widetilde{\mathcal{E}}^{1}_N(\bd w, \mathcal{L}_K)$ in Eq.~\eqref{eq:E1NL} is solely characterized by spectral constraints and we denote the corresponding set of admissible spectra by
\begin{align}\label{eq:Lambda-Kk}
\Lambda(\bd w, \mathcal{L}_K)  \equiv \Big\{\bd\lambda\in \RR^d\,|\,&\exists\gamma\in \widetilde{\mathcal{E}}^1_N(\bd w, \mathcal{L}_K), \pi\in\mathcal{S}_d: \nonumber\\
&\bd\lambda=\pi(\mathrm{spec}^\downarrow(\gamma))\Big\}\,.
\end{align}
Here, we use $\Lambda$ instead of $\Sigma$ (c.f., Eq.~\eqref{eq:Sigma-w}) to emphasise that $\Lambda(\bd w, \mathcal{L}_K) $ is in general \textit{not} a convex polytope for generic $\bd\lambda^{(i)}=(\mathcal{L}_K)_{\mu(i)}, i\in\{1, ..., |K|\}$. This can be understood by analogy with the moment polytope described by the generalized Pauli constraints, which is, in general, a convex polytope only for decreasingly ordered natural occupation numbers, but not for the full orbit of this polytope under the action of the symmetric group \cite{Klyachko2004, KL06, AK08}.

It is important to note that the mapping from 1RDMs to their vectors of eigenvalues is inherently non-linear. As a result, the spectrum of an 1RDM $\gamma \in \xxbar{\mathcal{E}}^1_N(\bd w,\mathcal{K}_K)$ cannot, in general, be expressed as a convex combination of the form $\pi(\mathrm{spec}^\downarrow(\gamma)) = \sum_{i=1}^r w_i \bd\lambda^{(i)}, \pi\in\mathcal{S}_d$. This non-linearity is a key feature of the spectral problem and highlights the relevance of the generalized Horn's problem \cite{Knutson, Bhatia, LP03} in this context. In particular, it motivates the connection of Horn's problem to our refined formulation of the one-body $N$-representability problem, which will be introduced in the next section as a framework for characterizing the set $\Lambda(\bd w, \mathcal{L}_K)$.

\subsubsection{Generalized Horn problem\label{sec:PSP}}

The generalized Horn problem, which deals with the sum of principal submatrices of a Hermitian matrix \cite{LP03}, extends the well-known Horn problem \cite{K98, KT99, Knutson, Bhatia} beyond the sum of two Hermitian matrices. Thus, Horn's problem corresponds to the case of $|K|=1$ fixed natural occupation number vectors (see Sec.~\ref{sec:recap}) and will be used to derive stricter bounds on the residuals of the $\bd w$-ensemble $N$-representability constraints in Sec.~\ref{sec:residuals}, where the residuals are defined as the difference between the right-hand side and the left-hand side of the $\bd w$-ensemble $N$-representability constraints. Additionally, the generalized Horn problem, with the original Horn problem as a special case, generalizes to settings with an arbitrary cardinality $|K|$ of the set $\mathcal{L}_K$ defined in Eq.~\eqref{eq:setNONKk}.

The generalized Horn problem raises the following question: Given $(m+1)$ decreasingly ordered vectors $\bd y^\downarrow, \bd x^{(i)\downarrow} \in \RR^d, i\in\{1, ..., m\}$, do there exist $(m+1)$ Hermitian matrices $ Y, X^{(i)} \in \RR^{d\times d}$ such that $Y = \sum_{i=1}^m X^{(i)}$?
Here, $\bd y^\downarrow = \mathrm{spec}^\downarrow(Y)$ and $\bd x^{(i)\downarrow} = \mathrm{spec}^\downarrow(X^{(i)})$.
Due to the convexity properties of moment polytopes (e.g., see Refs.~\cite{Kirwan, Guillemin-book}) the set of admissible $\bd y^\downarrow$ for some fixed $\bd x^{(i)\downarrow}, i\in\{1, ..., m\}$ is a convex polytope \cite{K98, KT99, Knutson, LP03}.
We denote this set of admissible $\bd y^\downarrow$ by $\Sigma_{\mathrm{H}}^\downarrow(\{\bd x^{(i)\downarrow}\}_{i=1}^m)$, where the subscript $\mathrm{H}$ stands for Horn's problem. As $\Sigma_{\mathrm{H}}^\downarrow(\{\bd x^{(i)\downarrow}\}_{i=1}^m)$ is a convex polytope, it is characterized by the intersection of finitely many hyperplanes. The corresponding halfspaces describe linear constraints on $\bd y^\downarrow$, which strongly depend on $d$ and are highly non-trivial already for the case of three matrices $Y, X^{(1)}, X^{(2)}$ as in the original Horn's problem. For a comprehensive derivation and discussion of the resulting linear constraints, we refer the reader to Refs.~\cite{K98, KT99, Knutson, Bhatia, LP03}.

Due to the Hermiticity of the 1RDM, we observe that the convex combination of 1RDM's
\begin{equation}\label{eq:ge}
\gamma = \sum_{i\in K}w_i \gamma^{(i)} +\widetilde{\gamma}
\end{equation}
and the individual terms in the above summation relate to the vectors $\bd y, \bd x^{(i)}$ in Horn's problem and its generalization according to $m=|K|+1$, $Y\equiv \gamma, X^{(i)}=w_{\mu^{-1}(i)}(\mathcal{L}_K)_{\mu(i)}$ for all $i\in \{1, ..., |K|\}$ and $X^{(|K|+1)} = \widetilde{\gamma}$, where only $\bd x^{(i)^\downarrow} \equiv w_i \bd\lambda^{(i)\downarrow}$ are known a priori.
This implies to consider the union of the set $\Sigma^\downarrow_\mathrm{H}(\{w_i\bd\lambda^{(i)\downarrow}\}_{i=1}^r)$, where $\bd\lambda^{(i)}$ are fixed for all $i\in K$, over all admissible $\bd\lambda^{(j)}\in \Sigma(\bd w_0)\forall \, j\in K^c$ according to
\begin{equation}\label{eq:union}
\Lambda_H^\downarrow(\bd w, \mathcal{L}_K) \equiv \bigcup_{\substack{\bd  x^{(|K|+1)\downarrow}\in \Sigma^\downarrow(\bd w_{K^c})}}\Sigma^\downarrow_\mathrm{H}(\{\bd x^{(i)\downarrow}\}_{i=1}^{|K|+1})\,,
\end{equation}
where $\bd w_0 = (1, 0, ...)$ and we defined the set
\begin{equation}\label{eq:TKc}
    \Sigma^\downarrow(\bd w_{K^c})\equiv \Big\{\bd\lambda\in \RR^d\,\Big\vert\,\exists\gamma\in \xxbar{\mathcal{E}}^1_N(\bd w_{K^c}): \bd\lambda = \mathrm{spec}^\downarrow(\gamma)\Big\}\,.
\end{equation}
We use the understanding and characterization of the set $\Lambda_H^\downarrow(\bd w, \mathcal{L}_K)$ in Eq.~\eqref{eq:union} in the following section to eventually characterize $\Lambda(\bd w, \mathcal{L}_K) $.

\subsubsection{Characterisation of $\Lambda(\bd w, \mathcal{L}_K) $\label{sec:Lambda}}

In this section, we use the solution to Horn's problem and its generalization for sums of $m$ Hermitian matrices, as discussed in Refs.~\cite{K98, KT99, Knutson, Bhatia, LP03}, to fully characterize $\Lambda(\bd w, \mathcal{L}_K)$ for generic $r$ and cardinality $|K|$ of the index set $K$. The first key result is the following lemma:
\begin{lem}\label{lem:1}
The set
\begin{equation}\label{eq:L-dec}
\Lambda^\downarrow(\bd w, \mathcal{L}_K) \equiv \left\{\bd\lambda^\downarrow\,|\,\bd\lambda\in \Lambda(\bd w, \mathcal{L}_K) \right\}
\end{equation}
of decreasingly ordered natural occupation number vectors $\bd\lambda$ is given by
\begin{equation}\label{eq:inter}
\Lambda^\downarrow(\bd w, \mathcal{L}_K) = \Lambda_H^\downarrow(\bd w, \mathcal{L}_K)\cap\Sigma(\bd w)\,.
\end{equation}
\end{lem}
In fact, every 1RDM $\gamma\in \widetilde{\mathcal{E}}^1_N(\bd w, \mathcal{L}_K)$ \eqref{eq:E1NL} can be written as a convex combination
\begin{equation}\label{eq:gtmp}
\gamma = \sum_{i\in K}w_i\gamma^{(i)} +\widetilde{\gamma}
\end{equation}
for some $\widetilde{\gamma}\in \xxbar{\mathcal{E}}^1_N(\bd w_{K^c})$. Together with Coleman's ensemble $N$-representability constraints \cite{C63} and the relaxed one-body $\bd w$-ensemble $N$-representability constraints characterizing $\Sigma(\bd w)$ \cite{SP21, LCLS21, CLLPPS23}, this yields the right-hand side of Eq.~\eqref{eq:inter}.

The union of convex sets, as in the definition of $\Lambda_H^\downarrow(\bd w, \mathcal{L}_K)$ in Eq.~\eqref{eq:union}, and its intersection with a convex set, is generally not convex.
Thus, it is essential to demonstrate that the set $\Lambda^\downarrow(\bd w, \mathcal{L}_K)$ indeed satisfies:
\begin{thm}\label{thm:1}
The set $\Lambda^\downarrow(\bd w, \mathcal{L}_K) $ is a convex polytope.
\end{thm}
\begin{proof}
We first prove that $\Lambda_H^\downarrow(\bd w, \mathcal{L}_K)$ as defined in Eq.~\eqref{eq:union} is convex. Since $\Sigma^\downarrow_\mathrm{H}(\{\bd x^{(i)\downarrow}\}_{i=1}^{|K|+1})$ is a convex polytope for any choice of $\bd x^{(|K|+1)\downarrow}\equiv\widetilde{\bd\lambda}^\downarrow\in \Sigma^\downarrow(\bd w_{K^c})$, it is characterized by
\begin{equation}
A \cdot\bd \lambda^\downarrow \leq C\cdot \widetilde{\bd\lambda}^\downarrow+ \sum_{i\in K}w_iB_i\cdot\bd\lambda^{(i)\downarrow}\,,
\end{equation}
where $A, B_i$ and $C$ are coefficient matrices.
The explicit $A, B_i,C$ are dependent on $r, |K|, d$ and follow from the solution to the principal submatrix problem for Hermitian matrices in Ref.~\cite{LP03}.

Then, for any pair of $\bd\lambda, \bd\lambda^\prime\in \Lambda_H^\downarrow(\bd w, \mathcal{L}_K) $ there exist respective $\widetilde{\bd\lambda}^\downarrow, \widetilde{\bd\lambda}^{\prime\downarrow}\in \Sigma(\bd w_{K^c})$ such that $\bd\lambda, \bd\lambda^\prime$ are in the respective polytope $\Sigma^\downarrow_\mathrm{H}(\{\bd x^{(i)\downarrow}\}_{i=1}^{|K|+1})$ (recall Eq.~\eqref{eq:union}). The convexity of the set $\Sigma(\bd w_{K^c})$ in Eq.~\eqref{eq:TKc} then implies that for any $\alpha\in [0,1]$,
\begin{equation}
\alpha \bd\lambda+(1-\alpha)\bd\lambda^\prime\in \Lambda_H^\downarrow(\bd w, \mathcal{L}_K)\,.
\end{equation}
Thus, $\Lambda_H^\downarrow(\bd w, \mathcal{L}_K)$ is convex.
Since the spectral polytope $\Sigma(\bd w)$ in Eq.~\eqref{eq:Sigma-w} is convex, also its intersection with the set $\Lambda_H^\downarrow(\bd w, \mathcal{L}_K)$ is convex. As both $\Sigma(\bd w)$ and $\Lambda_H^\downarrow(\bd w, \mathcal{L}_K)$ are not only convex but convex polytopes, also $\Lambda^\downarrow(\bd w, \mathcal{L}_K)$ is a convex polytope.
\end{proof}

\begin{figure}[tb]
\includegraphics[width=0.5\linewidth]{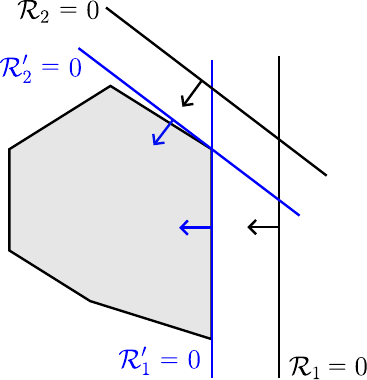}
\caption{Schematic illustration of the tightening of $\bd w$-ensemble $N$-representability constraints for two inequalities in the hyperplane representation of $\Sigma(\bd w)$. The vanishing residuals $R_j = 0$ (black) describing the two hyperplanes are shifted until they touch the boundary of the gray polytope leading to $R^\prime_j = 0$ (blue). (See text for more explanations.)
 \label{fig:residual-schematic}}
\end{figure}

Therefore, the characterization of the $\Lambda^\downarrow(\bd w, \mathcal{L}_K)$ in Eq.~\eqref{eq:Lambda-Kk} relies on the solution to the generalized Horn problem discussed in Sec.~\ref{sec:PSP}. These constraints, provided in Ref.~\cite{LP03}, apply to a general number $|K|+1$ of matrices in the weighted sums on the right-hand side of Eq.~\eqref{eq:gtmp}.

By Lemma \ref{lem:1} and Theorem \ref{thm:1}, a (generally non-minimal) hyperplane representation of $\Lambda^\downarrow(\bd w, \mathcal{L}_K)$ can be constructed by combining the hyperplane representations of $\Lambda_H^\downarrow(\bd w, \mathcal{L}_K)$ and $\Sigma(\bd w)$ \cite{SP21, LCLS21, CLLPPS23}. This combined representation also fully characterizes the non-convex set $\Lambda(\bd w, \mathcal{L}_K)$ defined in Eq.~\eqref{eq:Lambda-Kk}. Here, by combining two hyperplane representations, we simply mean to consider the collection of the inequalities arising from $\Lambda_H^\downarrow(\bd w, \mathcal{L}_K)$ and $\Sigma(\bd w)$. In general, this hyperplane representation of $\Lambda^\downarrow(\bd w, \mathcal{L}_K)$ will contain redundant inequalities, i.e., it will not be minimal, and a minimal hyperplane representation can be constructed.

\begin{figure*}[tb]
\includegraphics[width=\linewidth]{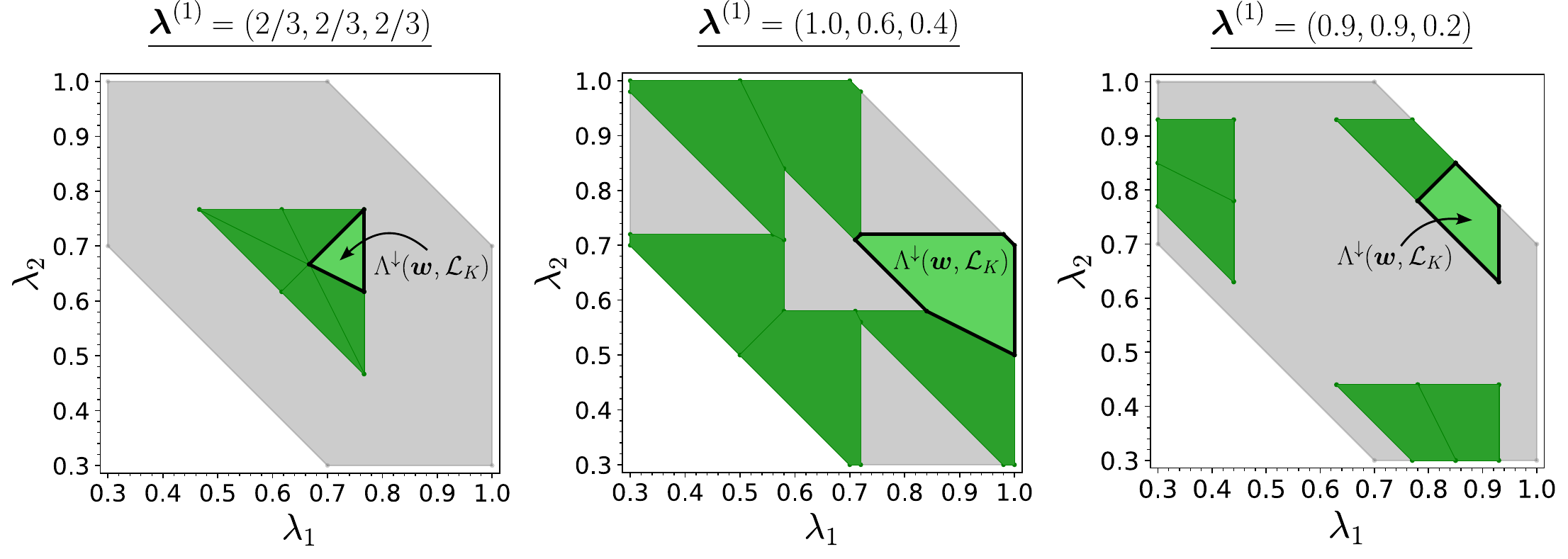}
\caption{Illustration of the non-convex polytope $\Lambda(\bd w, \bd\lambda^{(1)})$ (green) and their dependence on $\bd\lambda^{(1)}$ for $N=2, d=3$, $\bd w=(0.7,0.3,0, ...)$ and three different choices of $\bd\lambda^{(1)}$. The set $\Lambda(\bd w, \bd\lambda^{(1)})$ follows from $\Lambda^\downarrow(\bd w, \bd\lambda^{(1)})$ (light green) by considering all possible permutations of $\bd\lambda^\downarrow$. The spectral polytope $\Sigma(\bd w)$ is shown in gray and green. The value of $\lambda_3 = 2-\lambda_1-\lambda_2$ is fixed by the $l^1$-norm of $\bd\lambda$. (See text for more explanations.)
 \label{fig:Horn}}
\end{figure*}

Furthermore, if an inequality from the hyperplane representation of $\Sigma^\downarrow(\bd w)$ is not facet-defining for $\Lambda^\downarrow(\bd w, \mathcal{L}_K)$, it can be tightened by shifting it until it becomes tangent to the boundary of $\Lambda^\downarrow(\bd w, \mathcal{L}_K)$. Mathematically, this process involves adjusting the right-hand side of the inequality $\bd a^\mathrm{T}\cdot\bd\lambda^\downarrow\leq b$, where $\bd a$ is the normal vector of the corresponding hyperplane $\mathcal{R}_i \equiv b- \bd a^\mathrm{T}\cdot\bd\lambda^\downarrow =0$ and $b\in \RR$. The quantity $\mathcal{R}_i\geq 0$, representing the difference between the right-hand side and left-hand side of the inequality, is commonly referred to as the residual.

This situation is illustrated in Fig.~\ref{fig:residual-schematic}, where $\mathcal{R}_1$ and $\mathcal{R}_2$ represent two example residuals of the $\bd w$-ensemble $N$-representability constraints, expressed as $\mathcal{R}_i = b_i-\bd a_i^\mathrm{T} \cdot \bd \lambda^\downarrow  \geq 0$ (see also Sec.~\ref{sec:residuals} for examples). In contrast, the tighter residual $\mathcal{R}_1^\prime$ and its corresponding hyperplane (shown in blue), which is part of a hyperplane representation of $\Lambda_H^\downarrow(\bd w, \mathcal{L}_K)$, is facet-defining, i.e., it is tight for the convex polytope $\Lambda^\downarrow(\bd w, \mathcal{L}_K)$ shown in gray. Similarly, the hyperplane described by $R_2=0$ is shifted until it touches the boundary of $\Lambda^\downarrow(\bd w, \mathcal{L}_K)$.
In Sec.~\ref{sec:residuals}, we will make use of this tightening to derive more restrictive bounds on the residuals of the $\bd w$-ensemble $N$-representability constraints.

Furthermore, since both $\Sigma(\bd w_{K^c})$ \eqref{eq:TKc} and $\Sigma^\downarrow_\mathrm{H}(\{\bd x^{(i)\downarrow}\}_{i=1}^{|K|+1})$ in Eq.~\eqref{eq:union} are convex polytopes, the union over $\bd x^{(|K|+1)\downarrow}\equiv \widetilde{\bd\lambda}^\downarrow\in \Sigma(\bd w_{K^c})$ follows directly from a straightforward calculation.
To highlight this key aspect and illustrate the calculation of $\Lambda^\downarrow(\bd w, \mathcal{L}_K)$, we consider the example of $r=2$ nonzero weights $w_i$, $N=2$, and $d=3$. The sets $\Lambda(\bd w, \mathcal{L}_K)$ (green) and $\Lambda^\downarrow(\bd w, \mathcal{L}_K)$ (light green) for three different choices of $\bd\lambda^{(1)}$ are illustrated in Fig.~\ref{fig:Horn}. The original set $\Sigma(\bd w)$ of admissible $\bd w$-ensemble natural occupation number vectors, without prior knowledge of parts of the ensemble, is shown in grey. Due to the normalization of the 1RDM to the total particle number $N=2$, the third vector entry $\lambda_3 = 2 - \lambda_1 - \lambda_2$ is determined by the other two natural occupation numbers. Moreover, in Appendix \ref{app:cons}, we explicitly derive the hyperplane representation of the convex set $\Lambda^\downarrow(\bd w, \mathcal{L}_K)$, illustrating the application of Lemma \ref{lem:1}.

The in general non-convex  set $\Lambda(\bd w, \mathcal{L}_K)$ follows from taking the orbit of $\Lambda^\downarrow(\bd w, \mathcal{L}_K)$ under the action of the symmetric group $\mathcal{S}_d$, i.e., by considering all possible permutations of the vector entries of any $\bd\lambda\in\Lambda^\downarrow(\bd w, \mathcal{L}_K)$. In particular, Fig.~\ref{fig:Horn} stresses the dependence of the geometrical structure of the set $\Lambda(\bd w, \mathcal{L}_K)$ on the fixed natural occupation number vector $\bd\lambda^{(1)}$: for $\bd\lambda^{(1)}=(2/3, 2/3,2/3)$ (left panel), i.e., the maximally mixed scenario, $\Lambda(\bd w, \mathcal{L}_K)$ is indeed a convex polytope. On the contrary, the center of $\Sigma(\bd w)$ is not contained in $\Lambda(\bd w, \mathcal{L}_K)$ for less mixed $\bd\lambda^{(1)}$, as $\bd\lambda^{(1)}=(1, 0.6, 0.4), (0.9,0.9,0.2)$, in the middle and right panel. Moreover, also the number of facets and, thus, the number of facet-defining inequalities in a minimal hyperplane representation of $\Lambda^\downarrow(\bd w, \mathcal{L}_K)$ depends on $\bd\lambda^{(1)}$.

\subsubsection{Discussion\label{sec:discussion}}

The derivation of the set $\Lambda^\downarrow(\bd w, \mathcal{L}_K)$ for $|K|=1$ fixed natural occupation number vectors corresponds to the original Horn problem, which deals with the sum of two Hermitian matrices with fixed, decreasingly ordered vectors of eigenvalues, as explained in Sec.~\ref{sec:PSP}. Even for $|K|=1$, the number of linear constraints grows significantly with the dimension $d$ of the one-particle Hilbert space $\mathcal{H}_1$. For example, for $d = 2$, there are six constraints, and for $d = 3$, there are already twelve, as discussed in Refs.~\cite{Knutson, Bhatia}. For generic $|K|$, the number of constraints increases rapidly \cite{LP03}. Both $d = 2$ and $d = 3$ are far from typical settings in quantum chemistry or physics, where $d$ represents the dimension of the basis set or the number of lattice sites.

Furthermore, the hyperplane representations derived in Refs.~\cite{K98, KT99, Knutson, LP03} are in general not minimal, meaning that some linear inequalities are redundant. A minimal set of linear inequalities for Horn's problem can be obtained using so-called honeycombs as explained in Ref.~\cite{KTW04}. Moreover, due to the definition of $\Lambda_H^\downarrow(\bd w, \mathcal{L}_K)$, additional constraints are redundant, as we do not consider the generalized Horn problem for every set of Hermitian matrices.
At the same time, this does not resolve the issue of determining the inequalities that solve Horn's problem for generic $d$ in the first place. As discussed in the context of Fig.~\ref{fig:Horn}, the non-redundant linear constraints and, thus, the minimal hyperplane representation depend on $\bd\lambda^{(1)}$.

This leads to the following intermediate conclusion: By Lemma \ref{lem:1} and Theorem \ref{thm:1}, we provide a concrete scheme that can be implemented numerically to derive the linear constraints characterizing the set $\Lambda^\downarrow(\bd w, \mathcal{L}_K)$ and, consequently, $\Lambda(\bd w, \mathcal{L}_K)$. These constraints can be easily checked for a given natural occupation number vector $\bd\lambda$. However, it is the non-convexity of $\Lambda(\bd w, \mathcal{L}_K)$ that gives rise to the system-size dependent constraints from the generalized Horn problem in Ref.~\cite{LP03}, going beyond the relaxed one-body $\bd w$-ensemble $N$-representability constraints, which characterize $\Sigma(\bd w)$.

\begin{figure*}[tb]
\includegraphics[width=\linewidth]{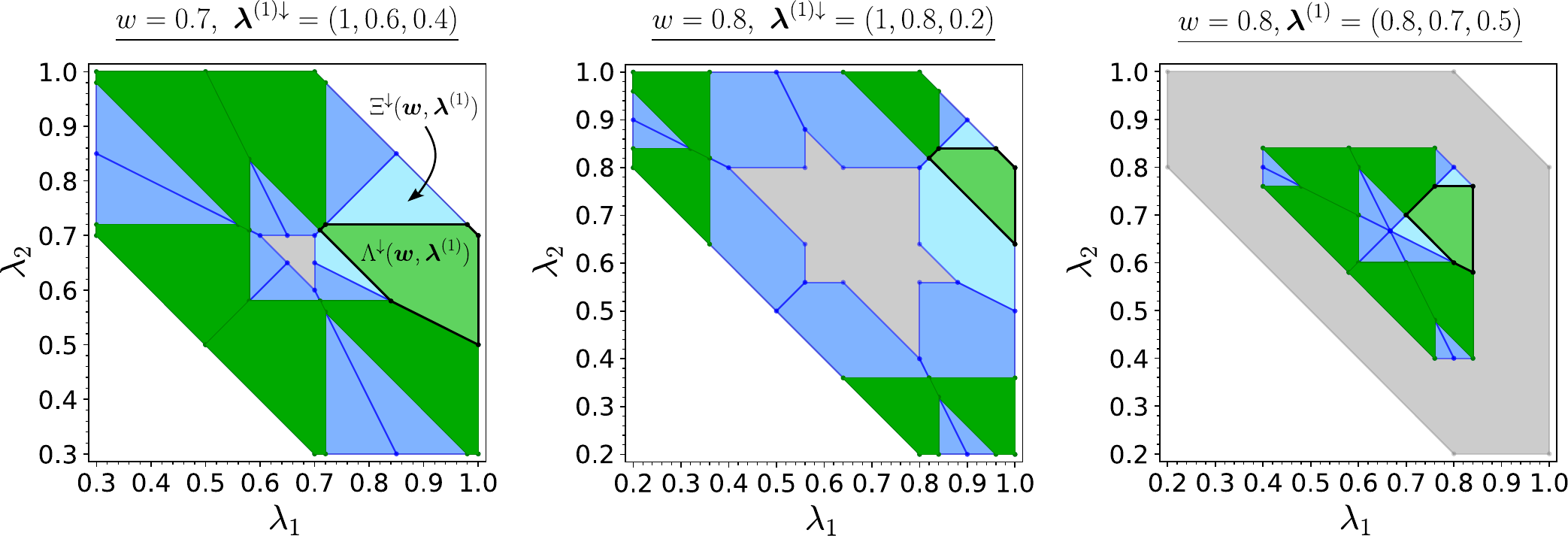}
\caption{Illustration of the spectral polytopes $\Sigma(\bd w, \bd\lambda^{(1)})$ (green) and $\Xi(\bd w, \bd\lambda^{(1)})$ (blue and green) for $r=2$ and three choices of $\bd w, \bd\lambda^{(1)}$. Both $\Lambda(\bd w, \bd\lambda^{(1)})$ and $\Xi(\bd w, \bd\lambda^{(1)})$ are subsets of the larger polytope $\Sigma(\bd w)$ (grey, green and blue), with $\Lambda(\bd w, \bd\lambda^{(1)})$ contained within $\Xi(\bd w, \bd\lambda^{(1)})$.
 \label{fig:residual_polytopes}}
\end{figure*}

\section{Improved bounds on residuals\label{sec:residuals}}

In this section, we derive tighter bounds on the residuals of the $\bd w$-ensemble $N$-representability constraints that define the set $\Sigma(\bd w)$ and present explicit examples of their application to small molecules. To this end, we begin with the simpler case of $r = 2$ non-vanishing weights in the $\bd w$-ensemble in Sec.~\ref{sec:residual-r2}. This serves as a primer for understanding the more general case of arbitrary $r \geq 3$, which is treated in Sec.~\ref{sec:residual-gen}.

\subsection{$\bd w$-ensembles with $r=2$\label{sec:residual-r2}}

The knowledge about $\mathcal{L}_K$ directly impacts the residuals of the inequalities constituting the hyperplane representation of $\Sigma(\bd w)$. Given $\Sigma(\bd w)$, it is always possible to determine a (generally non-minimal) hyperplane representation of $\Lambda^\downarrow(\bd w, \mathcal{L}_K)$ that includes hyperplanes with the same normal vectors as those in $\Sigma(\bd w)$. Thus, $\mathcal{L}_K$ dictates which of the hyperplanes --- from $\Sigma(\bd w)$ or $\Lambda^\downarrow(\bd w, \mathcal{L}_K)$ --- is tight.
Therefore, we obtain more restrictive residuals for the respective relaxed one-body $\bd w$-ensemble $N$-representability constraints. This procedure leads to a non-convex but easy-to-characterize outer approximation of $\Lambda(\bd w, \mathcal{L}_K)$.
In the following, we illustrate various theoretical concepts for $r=2$ as a proof-of-principle example. In fact, $r=2$ is particularly relevant for applications in functional theories, as it allows us to calculate the ground state energy and the ground state gap.

For $r = 2$, the set $\Sigma(\bd{w})$ is defined by Coleman's one-body $N$-representability constraints \cite{C63} (see Eq.~\eqref{eq:w-Nrepr-cons} in Appendix \ref{app:bounds-r2}), together with an additional $\bd{w}$-dependent constraint \cite{SP21, LCLS21, CLLPPS23} whose residual is given by
\begin{align}\label{eq:res-r2}
\mathcal{R}_N(\bd\lambda)&\equiv  N-1+w-\sum_{i=1}^N\lambda^\downarrow_i\geq 0\,.
\end{align}
Moreover, we assume that $\mathcal{L}_K=(\bd\lambda^{(1)})$, while the result for $\mathcal{L}_K=(\bd\lambda^{(2)})$ follows analogously.

As a key result, we prove in Appendix \ref{app:bounds-r2} that the sum of the $k$-largest natural occupation numbers $\lambda_i$ is bounded from above by
\begin{widetext}
\begin{equation}\label{eq:sum-r2-ub}
 \sum_{i=1}^k\lambda_i^\downarrow\leq \Omega(w, k, N, \mathcal{L}_K) = \begin{cases} w\sum_{i=1}^k\lambda_i^{(1)\downarrow}+(1-w)\min(k,N) \quad &\text{if } k\neq N\\
 \min\left(N-1+w, w\sum_{i=1}^k\lambda_i^{(1)\downarrow}+(1-w)N\right) \quad &\text{if } k=N\,,
 \end{cases}
\end{equation}
\end{widetext}
and from below by
\begin{equation}\label{eq:sum-r2-lb}
\sum_{i=1}^k\lambda_i^\downarrow\geq N - w\left(N - \sum_{i=1}^k\lambda_i^{(1)\downarrow}\right) - (1-w)\chi_{N,d}^{(k)}\,,
\end{equation}
where
\begin{equation}
\chi_{N,d}^{(k)} = \begin{cases}
N \quad &\text{if }d-1\geq N+k\,,\\
d-k-1\quad &\text{if }d-1< N+k
\end{cases}\,.
\end{equation}
The set of $\bd{\lambda}^\downarrow$ satisfying the constraints in Eqs.~\eqref{eq:sum-r2-ub} and \eqref{eq:sum-r2-lb} is denoted by $\Xi^\downarrow(\bd{w}, \bd\lambda^{(1)})$. Additionally, from the lower bounds in Eq.~\eqref{eq:sum-r2-lb}, we observe that the set $\Xi(\bd{w}, \bd\lambda^{(1)})$, obtained by considering all possible permutations of the spectra $\bd{\lambda} \in \Xi^\downarrow(\bd{w},\bd\lambda^{(1)})$, is non-convex. At the same time, the hyperplane representation of $\Xi^\downarrow(\bd{w},\bd\lambda^{(1)})$ in Eqs.~\eqref{eq:sum-r2-lb} and \eqref{eq:sum-r2-ub} is effectively independent of $N,d$. Furthermore, the spectral sets $\Xi(\bd{w}, \bd\lambda^{(1)})$ , $\Lambda(\bd{w}, \bd\lambda^{(1)})$ satisfy the inclusion relations
\begin{equation}
\Lambda\left(\bd w, \bd\lambda^{(1)}\right)\subseteq\Xi\left(\bd w, \bd\lambda^{(1)}\right)\subseteq \Sigma(\bd w)\,.
\end{equation}
We illustrate the sets $\Xi(\bd{w}, \mathcal{L}_K)$ (blue) , $\Lambda(\bd{w}, \mathcal{L}_K)$ (green) and $\Sigma(\bd w)$ (grey) for three different choices of $\bd w$ and $\bd\lambda^{(1)}$ in Fig.~\ref{fig:residual_polytopes}.

Through the knowledge of $\mathcal{L}_K$ constraints in Eqs.~\eqref{eq:sum-r2-ub} and \eqref{eq:sum-r2-lb} further imply more restrictive constraints on the corresponding residuals of the relaxed $\bd w$-ensemble $N$-representability constraints for all $k\in \{1, ..., d-1\}$ as detailed in Appendix \ref{app:bounds-r2}. In the following, we focus on the residual $\mathcal{R}_N(\bd\lambda)$ in Eq.~\eqref{eq:res-r2}, as it corresponds to the single $\bd w$-ensemble one-body $N$-representability constraint that is more restrictive than Coleman's ensemble $N$-representability conditions.
A brief calculation based on Eq.~\eqref{eq:sum-r2-ub} (see Appendix \ref{app:bounds-r2}) shows that the $\bd w$-dependent residual $\mathcal{R}_{\bd w}(\bd \lambda)$ is bounded from below by
\begin{equation}\label{eq:r2-lower-bound}
\mathcal{R}_N(\boldsymbol{\lambda}) \geq\max\left\{\mathcal{R}_{\min}^\mathrm{triv}, w\left(N-\sum_{i=1}^N \lambda_i^{(1)\downarrow}\right)-1+w\right\}\,,
\end{equation}
where $\mathcal{R}_{\min}^\mathrm{triv}=0$ represents the trivial lower bound, which holds independently of $\bd\lambda^{(1)}$.
Eq.~\eqref{eq:r2-lower-bound} shows that the non-trivial lower bound becomes tighter as the $l^1$-norm $\|\bd\lambda^{(1)}\|_1$ deviates further from $N$. For ground states of physical systems, this means the non-trivial bound becomes tighter as the ground state becomes more correlated.
In fact, a simple calculation shows that the new lower bound on $\mathcal{R}_N(\bd\lambda)$ in Eq.~\eqref{eq:r2-lower-bound} is more restrictive than the trivial lower bound $\mathcal{R}_N(\bd\lambda)\geq 0$ if
\begin{equation}\label{eq:crossing}
\sum_{i=1}^N \lambda_i^{(1)\downarrow}\leq N-\frac{1-w}{w}\,.
\end{equation}

We illustrate this behaviour using two chemical systems that exhibit strong correlation in certain parameter regimes: the hydrogen molecules $\text{H}_2$ and $\text{H}_4$, shown in Figs.~\ref{fig:h2} and \ref{fig:h4}, respectively. For both systems, we focus on the singlet sector and compute the ground and first excited states for various geometries using exact diagonalization in the cc-pvdz basis set. In the case of $\text{H}_4$, which consists of two $\text{H}_2$ dimers, the bond distance between the two hydrogen atoms in each dimer is fixed at $R_0 = 1.05 \mathring{\text{A}}$. The residual $\mathcal{R}_N(\bd \lambda)$ is shown in red, and the lower bound from Eq.~\eqref{eq:r2-lower-bound} is represented by the blue dashed lines.
While $\text{H}_2$ shows strong static correlation in the dissociation limit, $\text{H}_4$ is strongly correlated when the distance $R$ between the two $\text{H}_2$ dimers equals the fixed separation $R_0$. In particular, the new lower bound for the residual derived in Eq.~\eqref{eq:r2-lower-bound} is indeed tighter than the trivial bound $\mathcal{R}_{\min}^\mathrm{triv}=0$ in cases of strong correlation. The upper bounds on $\mathcal{R}_N(\bd \lambda)$ will be discussed next.

The residual $\mathcal{R}_N(\bd\lambda)$ is bounded from above by (see Appendix \ref{app:bounds-r2})
\begin{equation}\label{eq:r2-upper-bound}
\mathcal{R}_N(\boldsymbol{\lambda}) \leq\min\left\{\mathcal{R}_{\max}^\mathrm{triv}, \mathcal{R}_{\max}(\bd\lambda^{(1)})\right\}\,,
\end{equation}
where
\begin{widetext}
\begin{equation}
 \mathcal{R}_{\max}(\bd\lambda^{(1)})= \begin{cases}
N-1+w\left(1 -  \sum_{i=1}^N\lambda_i^{(1)\downarrow}  \right)   \quad &\text{if }d-1\geq 2N\,,\\
 d-2-N+w\left(2-d+2N -   \sum_{i=1}^N\lambda_i^{(1)\downarrow}  \right)          \quad &\text{if }d-1<2N
\end{cases}
\end{equation}
\end{widetext}
and the $\bd w$-dependent trivial upper bound $\mathcal{R}_{\max}^\mathrm{triv}$ on $\mathcal{R}(\bd\lambda)$ in Eq.~\eqref{eq:r2-upper-bound} is given by
\begin{equation}
 \mathcal{R}_{\max}^\mathrm{triv}=w-1+N\left(1-\frac{N}{d}\right)\,.
\end{equation}
As a consistency check, we observe that both the lower and upper bounds on $\mathcal{R}_N(\bd\lambda)$ in Eqs.~\eqref{eq:r2-lower-bound} and \eqref{eq:r2-upper-bound} reduce to the correct value $N-\sum_{i=1}^N\lambda_i^{(1)\downarrow}$ in the limit $w\to 1$.

\begin{figure}[t]
\includegraphics[width=\linewidth]{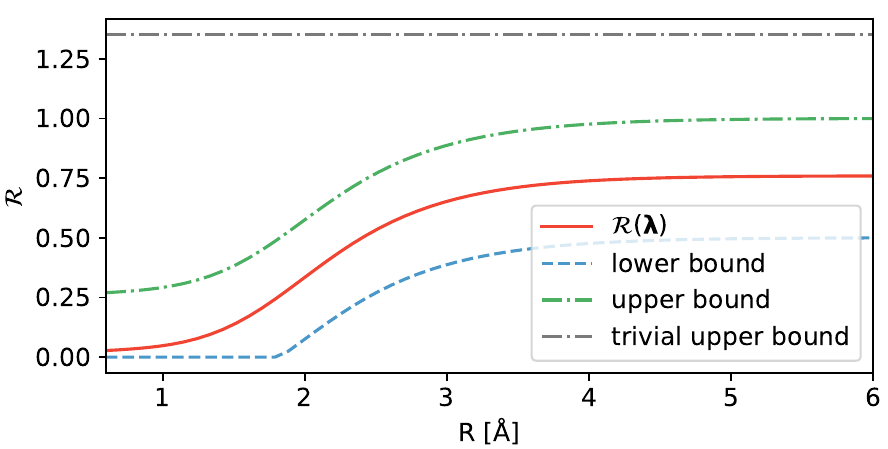}
\caption{Illustration of the residual $\mathcal{R}_N(\bd w)$ (red, solid) and its improved lower (blue, dashed) and upper (green, dashdotted) bound based on the knowledge of $\bd\lambda^{(1)}$ for $\text{H}_2$ in the singlet sector, $r=2, w=0.75$ in the cc-pvdz basis set. The trivial upper bound is shown in gray (dashdotted) and is always less restrictive than the new upper bound (see text for more details). \label{fig:h2}}
\end{figure}
\begin{figure}[t]
\includegraphics[width=\linewidth]{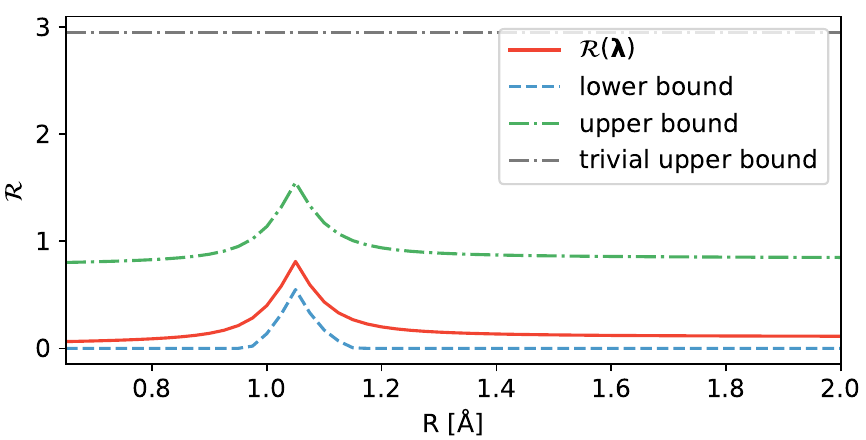}
\caption{Illustration of the residual $\mathcal{R}_N(\bd w)$ (red, solid) and its improved lower (blue, dashed) and upper (green, dashdotted) bound based on the information about the ground state natural occupation number vector $\bd\lambda^{(1)}$ for $\text{H}_4$ in the singlet sector, $r=2, w=0.75$ in the cc-pvdz basis set.
\label{fig:h4}}
\end{figure}

In contrast to the lower bound in Eq.~\eqref{eq:r2-lower-bound}, the non-trivial upper bound $\mathcal{R}_{\max}(\bd\lambda^{(1)})$ becomes tighter as $\|\bd\lambda^{(1)}\|_1$ deviates less from $N$, i.e., the weaker the correlation in the ground state of a Hamiltonian.
This behaviour is further illustrated for $\text{H}_2$ and $\text{H}_4$ in Figs.~\ref{fig:h2} and \ref{fig:h4}, where both $\mathcal{R}_{\max}^\mathrm{triv}$ (grey, dash-dotted) and $\mathcal{R}_{\max}(\bd\lambda^{(1)})$ (green, dash-dotted) are shown. In particular, we observe that $\mathcal{R}_{\max}(\bd\lambda^{(1)})$ is more restrictive than $\mathcal{R}_{\max}^\mathrm{triv}$ over the entire parameter regime for varying $R$ in both cases.

\subsection{Generic $r\geq 3$\label{sec:residual-gen}}

In this section, we outline the procedure for computing improved bounds on the residuals of the $\bd w$-ensemble $N$-representability constraints that characterize the set $\Sigma(\bd w)$ for a general value of $r$. Specifically, for $r \geq 3$, all up to $r$ natural occupation number vectors $\bd\lambda^{(i)}$ may be fixed, as discussed in Sec.~\ref{sec:natorb}.

When $|K| = r - 1$ natural occupation number vectors are fixed, we are left with only a single variable vector $\widetilde{\bd\lambda}\in \Sigma(\bd w_0)$. In this case, the derivation of the corresponding bounds on the residuals $\mathcal{R}_i$ proceeds analogously to the case $r = 2$. However, when $|K| < r - 1$, the union over all $\bd x^{(|K|+1)} \equiv \widetilde{\bd\lambda} \in \Sigma(\bd w_{K^c})$ in Eq.~\eqref{eq:union} reflects the fact that the remaining, non-fixed eigenvalue vectors arise from a $\bd w_{K^c}$-ensemble, as specified in Eqs.~\eqref{eq:ENwKc} and \eqref{eq:E1NwKc}.

As an illustrative example, we consider the case $r = 3$, where $\bd w = (w_1, w_2, w_3, 0, \ldots)$ with $w_3 = 1 - w_1 - w_2$. The general case for $r \geq 3$ follows analogously. Since the case $|K| = 2$ corresponds directly to the $r = 2$ scenario, we focus here on the case where only one, i.e., $|K| = 1$, natural occupation number vector is fixed. For concreteness, we assume this to be the one associated with the largest weight $w_1$. In this setting, we have $\bd w_{K^c} = (w_2, w_3, 0, \ldots)$ and $\widetilde{\bd\lambda} \in \Sigma(\bd w_{K^c})$.
To compute the union over all $\widetilde{\bd\lambda} \in \Sigma(\bd w_{K^c})$, we must first characterize the set $\Sigma(\bd w_{K^c})$, which can be done analogously to the characterization of $\Sigma(\bd w)$ presented in Refs.~\cite{SP21, LCLS21}. In particular, a straightforward calculation shows that $\Sigma(\bd w_{K^c})$ corresponds to the permutohedron generated by the vertex
\begin{equation}\label{eq:vwKc}
    \bd v = (\underbrace{w_2+w_3, ..., w_2+w_3}_{N-1}, w_2, w_3, 0, ...)\,,
\end{equation}
and thus given by
\begin{equation}\label{eq:Sigmawc-r3}
    \Sigma(\bd w_{K^c}) = \mathrm{conv}\left(\left\{\pi(\bd v)\,|\,\pi\in \mathcal{S}_d\right)\right\}\,.
\end{equation}
In particular $\|\bd v\|_1\neq N$, meaning that the $\bd w_{K^c}$ correspond to linear operators not normalized to unity (recall Eqs.~\eqref{eq:ENwKc} and \eqref{eq:EbNwKc}).

For $r=3$, there are three non-trivial $\bd w$-dependent constraints characterizing the set $\Sigma(\bd w)$ \cite{SP21, LCLS21, CLLPPS23}. To illustrate the general procedure for deriving the refined bounds on the corresponding residuals, we pick as an example the residual $\mathcal{R}_N$ in Eq.~\eqref{eq:res-r2}, which is also present in the case of $r=2$. Then, following an argument analogous to the derivation of Eq.~\eqref{eq:sum-r2-ub}, the sum of the $N$ largest eigenvalues is bounded from above by
\begin{widetext}
\begin{align}\label{eq:res-r3-up}
    \sum_{i=1}^N\lambda_i^\downarrow\leq &\min\Big(N-1+w_1, w_1\sum_{i=1}^N\lambda^{(1)\downarrow}_i+  (N-1)(w_2+w_3)+w_2  \Big)\,,
\end{align}
\end{widetext}
where, on the right-hand side, we have used the bound
\begin{equation}
    \sum_{i=1}^N\widetilde{\lambda}_i^\downarrow\leq (N-1)(w_2+w_3)+w_2
\end{equation}
which follows from Rado's theorem \cite{R52} and the hyperplane representation of $\Sigma(\bd w_{K^c})$ Eq.~\eqref{eq:Sigmawc-r3}.
The corresponding lower bound on $\sum_{i=1}^N \lambda_i^\downarrow$ can be derived in an analogous manner and, together with Eq.~\eqref{eq:res-r3-up}, yields improved bounds on $\mathcal{R}_N$ for $r = 3$. Similarly, tighter bounds on the additional residuals constraining the set $\Sigma(\bd w)$ can be obtained, though we leave the explicit calculations to the reader.

\section{Convexification and application to lattice DFT \label{sec:spectral}}

In this section, we discuss the convex hull of $\Xi(\bd w, \bd\lambda^{(1)})$. Moreover, we illustrate that the respective set characterizes the set of admissible densities in lattice DFT for $r=2$.

We first observe that by definition the convex hull of $\Xi(\bd w, \bd\lambda^{(1)})$,
\begin{equation}\label{eq:Sigma-Xi}
\Sigma(\bd w, \bd\lambda^{(1)})\equiv \mathrm{conv}\left(\Xi\left(\bd w, \bd\lambda^{(1)}\right)\right)\,,
\end{equation}
is fully characterized by the upper bounds on the sum of the $k$-largest natural occupation numbers in Eq.~\eqref{eq:sum-r2-ub} for all $k\in \{1, ..., d-1\}$.

The natural variable in lattice DFT is the lattice site occupation number vector $\bd n = \mathrm{diag}(\gamma)$ \cite{GS86, SGN95, IH10, XCTK12, SP14, Coe19, PvL21, PvL23, PvL24}. Therefore, the goal is to determine the set of $\bd n$ in ensemble lattice DFT for excited states (lattice EDFT) \cite{GOK88c, YPBU17, Fromager2020-DD, Loos2020-EDFA, Cernatic21, Yang21, GKGGP23, SKCPJB24, CPSF24} compatible with $\mathcal{L}_k$ for $r = 2$.
The Schur-Horn theorem \cite{Schur1923, Horn54} implies that any lattice site occupation number vector $\bd  n$ is majorized by the vector $\bd \lambda$ of eigenvalues of the 1RDM leading to
\begin{equation}\label{eq:ub-density}
 \sum_{i=1}^kn_i^\downarrow\leq  \sum_{i=1}^k\lambda_i^\downarrow\leq \Omega(w, k, N, \mathcal{L}_K)
\end{equation}
with $\Omega(w, k, N, \mathcal{L}_K)$ given by Eq.~\eqref{eq:sum-r2-ub}. It follows that for fixed $\bd\lambda^{(1)}$, the lattice site occupation number vector $\bd n$ in an ensemble DFT calculation for $r=2$ must be an element of the set $\Sigma(\bd w, \bd\lambda^{(1)})$.
To apply Eq.~\eqref{eq:ub-density} in ensemble lattice EDFT, one may exploit the knowledge of the ground state natural occupation number vector $\bd \lambda^{(1)}$ to restrict the set of admissible lattice site occupation number vectors $\bd n$ further. As a consequence of Eq.~\eqref{eq:r2-upper-bound}, Eq.~\eqref{eq:ub-density} will be most restrictive if the ground state has strong static correlation. In addition, one may exploit the freedom in choosing the weight vector $\bd w$ to make the set as restrictive as possible for a given $\bd\lambda^{(1)}$.
For instance, the additional constraints in Eq.~\eqref{eq:ub-density} will be particularly important when ground state functionals are employed to model also the excited states, as it is commonly the case in ensemble DFT.

We now provide a more detailed discussion of the geometric properties of $\Sigma(\bd w, \bd\lambda^{(1)})$. Since $\Sigma(\bd w)$ is a permutohedron for $r=2$, also the set $\Sigma(\bd w, \bd\lambda^{(1)})$ is a permutohedron,
\begin{equation}
\Sigma(\bd w, \bd \lambda^{(1)}) = \mathrm{conv}\left(\left\{\pi(\widetilde{\bd v})\,|\,\pi\in\mathcal{S}_d \right\}\right)\,.
\end{equation}
The generating vertex $\widetilde{\bd v}$ of $\Sigma(\bd w, \bd \lambda^{(1)})$ can be obtained from the hyperplane representation of $\Sigma(\bd w, \bd \lambda^{(1)})$ through Eq.~\eqref{eq:sum-r2-ub}. This eventually leads to (see Appendix \ref{app:v-rep})
\begin{align}
\widetilde{v}_i &= w\lambda_i^{(1)\downarrow}+(1-w)\quad \text{ for }1\leq i\leq N-1\,,\nonumber\\
\widetilde{v}_N &=\min\Big(w\lambda_N^{(1)\downarrow}+(1-w), N-1+w -\sum_{j<N}\widetilde{v}_j\Big)\,,\nonumber\\
\widetilde{v}_{N+1} &= w\lambda_{N+1}^{(1)\downarrow} + \widetilde{x}\nonumber\\
\widetilde{v}_{i} &= w\lambda_i^{(1)\downarrow}\quad \forall i \in \{N+2, d\}\,,
\end{align}
where we defined
\begin{equation}
\widetilde{x}\equiv  w\lambda_N^{(1)\downarrow} +(1-w)-\widetilde{v}_N\geq 0  \,.
\end{equation}
As a consistency check, we note that Rado's theorem \cite{R52} produces the inequalities in Eq.~\eqref{eq:sum-r2-ub}.
Moreover, Eq.~\eqref{eq:sum-r2-ub} reveals that unless all $(N-1)$ largest entries of $\bd\lambda^{(1)}$ equal one (recall Fig.~\ref{fig:residual_polytopes} for $N=2$), the new exclusion principle constraints defining the spectral set $\bd\lambda \in \Sigma(\bd w, \bd \lambda^{(1)})$ are always stricter than the $\bd w$-constraints that define $\Sigma(\bd w)$.

\begin{figure}[tb]
\includegraphics[width=0.85\linewidth]{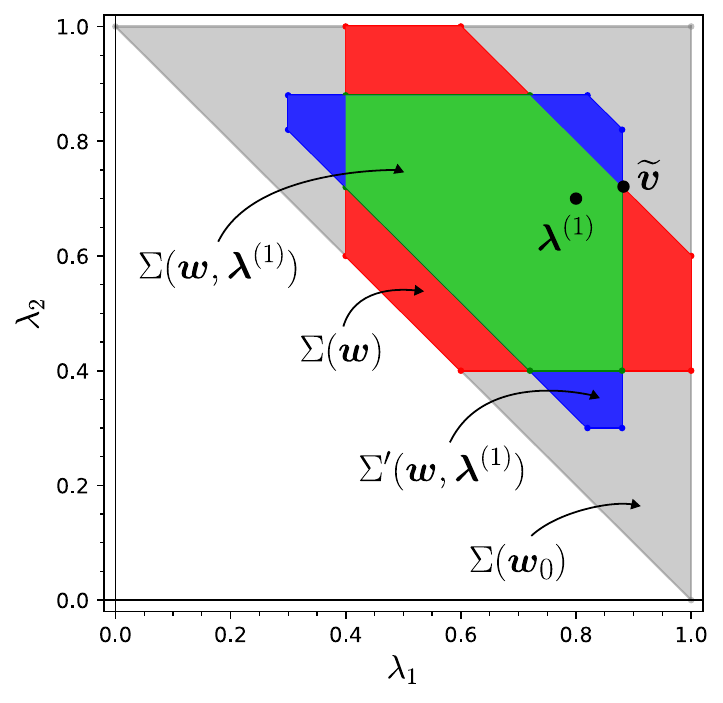}
\caption{Illustration of the construction of the spectral set $\Sigma(\bd w, \bd\lambda^{(1)})$ \eqref{eq:Sigma-Xi} (green) for $N=2, d=3$ such that $\lambda_3=N-\lambda_1-\lambda_2$. The spectral set $\Sigma(\bd w, \bd \lambda^{(1)})$ for fixed $\bd\lambda^{(1)} =(0.8,0.7,0.5)$ and $w=0.6$ follows from the intersection of $\Sigma^\prime(\bd w, \bd \lambda^{(1)})$ (blue) and $\Sigma(\bd w)$ (red). The generating vertex of $\Sigma(\bd w, \bd\lambda^{(1)})$ is denoted by $\widetilde{\bd v}$, and the grey polytope illustrates the Pauli simplex (see text for more explanations).
 \label{fig:polytopes}}
\end{figure}

Furthermore, using the hyperplane representation of the spectral set $\Sigma(\bd{w}, \bd{\lambda}^{(1)})$, it follows from Eq.~\eqref{eq:Sigma-Xi} that $\Sigma(\bd{w}, \bd{\lambda}^{(1)})$ is given by the intersection.
\begin{equation}\label{eq:sigma-non-r2}
\Sigma(\bd w,  \bd\lambda^{(1)})\equiv \Sigma(\bd w)\cap \Sigma^\prime(\bd w, \bd\lambda^{(1)})\,,
\end{equation}
where
\begin{equation}\label{eq:S-prime-r2}
\Sigma^\prime(\bd w,\bd\lambda^{(1)})\equiv w P_{\bd\lambda^{(1)}}+ (1-w) P_{\bd v^\mathrm{HF}}
\end{equation}
with $\bd v^\mathrm{HF} = (1, ..., 1, 0, ...)$. Since $\Sigma^\prime(\bd w, \bd\lambda^{(1)})$ is convex as the Minkowski sum of two convex polytopes, $\Sigma(\bd w, \bd\lambda^{(1)})$, being the intersection of two polytopes, is also convex in agreement with Eqs.~\eqref{eq:Sigma-Xi} and \eqref{eq:sigma-non-r2}.

We illustrate Eq.~\eqref{eq:sigma-non-r2} and \eqref{eq:S-prime-r2} for $r=2,N=2,d=3$ and $\bd\lambda^{(1)} = (0.8,0.7,0.5), w =0.6$ in Fig.~\ref{fig:polytopes} (see also the right panel of Fig.~\ref{fig:residual_polytopes}). The set $\Sigma^\prime(\bd w, \bd\lambda^{(1)})$ (blue) in Eq.~\eqref{eq:S-prime-r2} follows from the Minkowski sum (i.e., element-wise sum of sets \cite{Ziegler1994}) of the permutohedron of $\bd\lambda^{(1)}$ with the Pauli simplex $\Sigma(\bd w_0)$ (gray) (recall that $\bd w_0 = (1, 0, ...)$). According to Eq.~\eqref{eq:sigma-non-r2}, $\Sigma(\bd w, \bd\lambda^{(1)})$ (green) eventually follows as the intersection of $\Sigma(\bd w)$ (red) and $\Sigma^\prime(\bd w, \bd\lambda^{(1)})$ (blue).

\section{Conclusions and outlook\label{sec:concl}}

In this work, we systematically addressed the refined $N$-representability problem for $\boldsymbol{w}$-ensemble one-body reduced density matrices (1RDMs), focusing on scenarios where partial information about the ensemble—such as full 1RDMs or their natural occupation numbers—is known \textit{a priori}. Such situations naturally arise in $\boldsymbol{w}$-RDMFT calculations~\cite{SP21, LCLS21, LS23-sp, LS23-njp}, where excited states are computed sequentially. Ensembles also emerge in systems with degenerate eigenstates and in open quantum systems. Extensions of the present framework to these broader contexts may be explored in future work.

We showed that incorporating known 1RDMs enables the formulation of a refined $\boldsymbol{w}$-ensemble one-body $N$-representability problem, which by construction imposes stricter constraints on the admissible set of 1RDMs than the original formulation~\cite{SP21, LCLS21, LS23-sp, LS23-njp}. However, this refined problem is inherently complex, as it is not a purely spectral problem and the admissible set of 1RDMs is non-convex. To address this, we introduced a hierarchy of systematic relaxations, ultimately yielding a spectral set defined by eigenvalue constraints that are practically characterizable. These developments provide a tractable and physically meaningful approach for incorporating refined $N$-representability constraints into functional theories.

We further investigated the geometric structure of the resulting non-convex spectral set and derived a hyperplane representation in terms of a finite number of linear inequalities on the natural occupation numbers. The physical significance of these system-size-independent constraints was supported by both analytic derivation and numerical illustration in molecular systems exhibiting varying degrees of correlation.

A key application of this work lies in $\boldsymbol{w}$-RDMFT~\cite{SP21, LCLS21, LS23-sp, LS23-njp} and ensemble density functional theory (EDFT) for excited states~\cite{GOK88b, GOK88c, Loos2020-EDFA, Fromager2020-DD, Cernatic21, GK21, GL23, SKCPJB24, Fromager24-EDFT, CLSF24}. In $\boldsymbol{w}$-RDMFT, excited states are typically computed through a sequence of calculations with varying weight vectors $\boldsymbol{w}$, each containing an increasing number of non-zero components. At each step, the natural occupation number vectors of previously computed excited-state 1RDMs can be used to impose additional constraints, thereby restricting the optimization to a smaller, more physically relevant subset of ensemble 1RDMs. This is particularly advantageous when using approximate $\boldsymbol{w}$-RDMFT functionals, where such constraints may significantly enhance accuracy and stability. In particular, when computing energy gaps, this strategy allows one to leverage ground-state information already available from existing RDMFT functionals~\cite{GPB05, Piris19, Piris21-1, P21, LKO22, GBM22, SNF22, GBM23, GBM24, CG24, YS24, VMSS24}.

Current EDFT implementations often rely on ground-state functionals adapted to $\boldsymbol{w}$-ensembles due to practical limitations. The incorporation of additional information, such as natural occupation numbers from RDMFT or wave function methods, offers a path toward improved excited-state predictions within this framework.

These constraints can also be incorporated into variational two-electron reduced density matrix (2-RDM) methods~\cite{Mazziotti.2004, Mazziotti.2011mwm, Mazziotti.2020z0p, Knight.2022, Schouten.2023, Xie.20222n} to provide approximate excited-state solutions. As in the case of the 1RDM theories, the constraints may permit a series of variational 2-RDM calculations with an increasing number of nonzero spectral weights $w_i$, from which an approximation to the excited states and their 2-RDMs can be computed.

Beyond applications in functional and variational methods, partial information about quantum ensembles—and the additional constraints it imposes on reduced density matrices—may also prove valuable in reduced density matrix tomography~\cite{XLKL17, ZNM21, Mazziotti.2020z0p, Knight.2022, Schouten.2023, Xie.20222n} and in machine learning approaches for quantum systems~\cite{SPTP23, HPS24}. By integrating physically grounded constraints derived from ensemble theory, these data-driven or inference-based methods may achieve improved reliability, interpretability, and performance in future applications.

\begin{acknowledgments}
We acknowledge financial support from the U.S. National Science Foundation (NSF) Grant No. CHE-2155082 (D.M.), the German Research Foundation (Grant SCHI 1476/1-1) (J.L., C.S.), the U.S. NSF Graduate Research Fellowship Program under Grant No. 2140001 (I.A.), and the International Max Planck Research School for Quantum Science and Technology (IMPRS-QST) (J.L.). The project/research is also part of the Munich Quantum Valley, which is supported by the Bavarian state government with funds from the Hightech Agenda Bayern Plus.
\end{acknowledgments}
\appendix

\section{Refined non-relaxed $\bd w$-ensemble $N$-representability constraints for $d=3$ from Horn's problem\label{app:cons}}

In this section, we show as an example for $d=3$ how the refined $\bd w$-ensemble $N$-representability constraints characterizing the non-convex set $\Lambda(\bd w, \mathcal{L}_K)$ are obtained according to Lemma \ref{lem:1}. These will be the constraints that characterize the set $\Lambda^\downarrow(\bd w, \mathcal{L}_K)$ for the example in Fig.~\ref{fig:Horn} in Sec.~\ref{sec:natorb}.

We first recall Horn's problem discussed in Sec.~\ref{sec:PSP}, i.e., the principal submatrix problem for three Hermitian matrices. Horn's problem asks the following question: given three decreasingly ordered vectors $\bd x^\downarrow, \bd y^\downarrow, \bd z^\downarrow\in \RR^d$, do there exist three Hermitian matrices $X, Y, Z\in \RR^{d\times d}$ such that $Z=X+Y$?
The set of admissible $\bd z^\downarrow$ for fixed $\bd x^\downarrow, \bd y^\downarrow$ is a convex polytope denoted by $\Sigma_{\mathrm{H}}^\downarrow(\bd x, \bd y)$, where the subscript $\mathrm{H}$ stands for Horn's problem.
In our case the ensemble decomposition
\begin{equation}
\gamma = w\gamma^{(1)}+(1-w)\gamma^{(2)}
\end{equation}
is used to assign $Z\equiv \gamma, X=w\gamma^{(1)}, Y=(1-w)\gamma^{(2)}$, where only $\bd x^\downarrow \equiv \bd\lambda^{(1)\downarrow}$ is fixed. For $\bd y^\downarrow\equiv \bd\lambda^{(2)\downarrow}$ we consider all $\bd\lambda^{(2)\downarrow}\prec \bd v^\mathrm{HF}$ which are majorized by the Hartree-Fock natural occupation number vector $\bd v^\mathrm{HF} =  (1, ..., 1, 0, ..., 0)$.

The goal of this section is to find the hyperplane representation of the set
\begin{equation}\label{eq:app-ld}
\Lambda^\downarrow(\bd w, \mathcal{L}_K) = \Lambda_H^\downarrow(\bd w, \mathcal{L}_K)\cap\Sigma(\bd w)
\end{equation}
of decreasingly ordered natural occupation number vectors $\bd\lambda\in \RR^d$ as defined in Eq.~\eqref{eq:inter} with $\Lambda_H^\downarrow(\bd w, \mathcal{L}_K)$ given by Eq.~\eqref{eq:union} in the main text. A hyperplane representation of $\Sigma(\bd w)$ is known from Refs.~\cite{SP21, LCLS21, CLLPPS23}. To calculate the intersection of the two sets on the right hand-side of Eq.~\eqref{eq:app-ld}, we have to determine (recall Eq.~\eqref{eq:union})
\begin{equation}\label{eq:app-union}
\Lambda_H^\downarrow(\bd w, \mathcal{L}_K) \equiv \bigcup_{\substack{\bd\lambda^{(2)^\downarrow}\in\Sigma^\downarrow(\bd w_0)}}\Sigma^\downarrow_\mathrm{H}(\{w_i\bd\lambda^{(i)\downarrow}\}_{i=1, 2})\,.
\end{equation}

For $d=3$ as discussed in this section, the solution to Horn's problem which characterizes the convex polytope $\Sigma^\downarrow_\mathrm{H}(\{w_i\bd\lambda^{(i)\downarrow}\}_{i=1, 2})$ is given by the twelve inequalities \cite{Bhatia, Fulton00, Knutson}
\begin{align}
&\lambda_1^\downarrow \leq w\lambda_1^{(1)\downarrow}+(1-w)\lambda_1^{(2)\downarrow}\,,\,\,\lambda_2^\downarrow \leq w\lambda_1^{(1)\downarrow}+(1-w)\lambda_2^{(2)\downarrow}\,,\nonumber\\
& \lambda_2^\downarrow \leq w\lambda_2^{(1)\downarrow}+(1-w)\lambda_1^{(2)\downarrow}\,,\,\,\lambda_3^\downarrow \leq w\lambda_1^{(1)\downarrow}+(1-w)\lambda_3^{(2)\downarrow}\,,\nonumber\\
& \lambda_3^\downarrow \leq w\lambda_3^{(1)\downarrow}+(1-w)\lambda_1^{(2)\downarrow}\,,\,\,
 \lambda_3^\downarrow \leq w\lambda_2^{(1)\downarrow}+(1-w)\lambda_2^{(2)\downarrow}\,,\nonumber\\
&\lambda_1^\downarrow + \lambda_2^\downarrow \leq w(\lambda_1^{(1)\downarrow}+\lambda_2^{(1)\downarrow})+(1-w)(\lambda_1^{(2)\downarrow}+\lambda_2^{(2)\downarrow})\,,\nonumber\\
&\lambda_1^\downarrow + \lambda_3^\downarrow \leq w(\lambda_1^{(1)\downarrow}+\lambda_3^{(1)\downarrow})+(1-w)(\lambda_1^{(2)\downarrow}+\lambda_2^{(2)\downarrow})\,,\nonumber\\
&\lambda_2^\downarrow + \lambda_3^\downarrow \leq w(\lambda_2^{(1)\downarrow}+\lambda_3^{(1)\downarrow})+(1-w)(\lambda_1^{(1)\downarrow}+\lambda_2^{(1)\downarrow})\,,\nonumber\\
&\lambda_1^\downarrow + \lambda_3^\downarrow \leq w(\lambda_1^{(1)\downarrow}+\lambda_2^{(1)\downarrow})+(1-w)(\lambda_1^{(2)\downarrow}+\lambda_3^{(2)\downarrow})\,,\nonumber\\
&\lambda_2^\downarrow + \lambda_3^\downarrow \leq w(\lambda_1^{(1)\downarrow}+\lambda_2^{(1)\downarrow})+(1-w)(\lambda_2^{(2)\downarrow}+\lambda_3^{(2)\downarrow})\,,\nonumber\\
&\lambda_2^\downarrow + \lambda_3^\downarrow \leq w(\lambda_1^{(1)\downarrow}+\lambda_3^{(1)\downarrow})+(1-w)(\lambda_1^{(2)\downarrow}+\lambda_3^{(2)\downarrow})\,.
\end{align}
Due to the linearity of ensemble one-body $N$-representability constraints, the union over all $\bd\lambda^{(2)\downarrow}\in\Sigma(\bd w_0)$ follows from a straightforward calculation for fixed value of the total particle number $N$.

Then, the hyperplane representation of the set $\Lambda_H^\downarrow(\bd w, \mathcal{L}_K)$ in Eq.~\eqref{eq:app-union} for $N=2$ follows as
\begin{align}\label{eq:app-cons}
\lambda_1^\downarrow &\leq w\lambda_1^{(1)\downarrow}+1-w\,,\nonumber\\
 \lambda_2^\downarrow &\leq w\lambda_2^{(1)\downarrow}+1-w\,,\quad\lambda_3^\downarrow \leq w\lambda_1^{(1)\downarrow}+\frac{2(1-w)}{3}\,,\nonumber\\
 \lambda_3^\downarrow &\leq w\lambda_3^{(1)\downarrow}+1-w\,,\nonumber\\
\lambda_1^\downarrow + \lambda_3^\downarrow &\leq w(\lambda_1^{(1)\downarrow}+\lambda_2^{(1)\downarrow})+\frac{3(1-w)}{2}\,,\nonumber\\
\lambda_2^\downarrow + \lambda_3^\downarrow &\leq w(\lambda_1^{(1)\downarrow}+\lambda_2^{(1)\downarrow})+\frac{4(1-w)}{3}\,,\nonumber\\
\lambda_2^\downarrow + \lambda_3^\downarrow &\leq w(\lambda_1^{(1)\downarrow}+\lambda_3^{(1)\downarrow})+\frac{3(1-w)}{2}\,.
\end{align}
Together with the relaxed $\bd w$-ensemble one-body $N$-representability constraints describing $\Sigma(\bd w)$ the constraints in Eq.~\eqref{eq:app-cons} yield a hyperplane representation of the set $\Lambda^\downarrow(\bd w, \bd\lambda^{(1)})$. Moreover, a minimal hyperplane representation can be deduced based on $\bd\lambda^{(1)}$.

\section{Derivation of the upper and lower bounds on the residuals for $r=2$\label{app:bounds-r2}}

In this section, we derive the upper and lower bounds for the residuals of the relaxed one-body $\bd w$-ensemble $N$-representability constraints that characterize the spectral set $\Sigma(\bd w)$ defined in Eq.~\eqref{eq:Sigma-w}.

We first recall that a hyperplane representation of $\Sigma(\bd w)$ is given by \cite{SP21, LCLS21, CLLPPS23}
\begin{align}\label{eq:w-Nrepr-cons}
\sum_{i=1}^k\lambda_i^\downarrow&\leq k\quad \forall k\in\{1, ..., N-1\}\,, \nonumber\\
\sum_{i=1}^N\lambda_i^\downarrow&\leq N-1+w\nonumber\\
\sum_{i=1}^k\lambda_i^\downarrow&\leq N\quad \forall k\in\{N+1, ..., d-1\}\,,\nonumber\\
\sum_{i=1}^d\lambda_i^\downarrow&=1\,.
\end{align}
Any linear inequality in Eq.~\eqref{eq:w-Nrepr-cons} can be written in the form $\bd a^\mathrm{T}\cdot\bd\lambda^\downarrow\leq b$.
The residuals of these linear inequalities are defined as the difference between the right-hand side and the left-hand side, i.e., $\mathcal{R}_i(\bd\lambda) \equiv b-\bd a^\mathrm{T}\cdot\bd\lambda^\downarrow\geq 0$. For the single $\bd w$-dependent constraint in Eq.~\eqref{eq:w-Nrepr-cons} this leads to Eq.~\eqref{eq:res-r2} in Sec.~\ref{sec:residuals}.

Ky Fan showed in 1949, that a necessary linear inequality for the solution to Horn's problem is given by \cite{Fulton00, Bhatia} (translated into our setting)
\begin{equation}\label{eq:KyFan}
\sum_{i=1}^k \lambda_i^\downarrow\leq \sum_{i=1}^k\left(w\lambda_i^{(1)\downarrow} +(1-w)\lambda_i^{(2)\downarrow}\right)\,\,\forall k\leq d\,.
\end{equation}
The above inequality immediately implies the sum of the $k$ largest natural occupation numbers $\lambda_i$ is bounded from below according to
\begin{equation}
\sum_{i=1}^k\lambda_i^\downarrow \geq w\sum_{i=1}^k\lambda_i^{(1)\downarrow} +(1-w)\min(k,N)\,.
\end{equation}
In addition, we have the following upper bound for the smallest $d-k-1$ entries of $\bd\lambda$,
\begin{equation}\label{eq:KyFan-inverse}
\sum_{i=k+1}^d\lambda_i^\downarrow\leq w\sum_{i=k+1}^d\lambda_i^{(1)\downarrow}+(1-w)\sum_{i=1}^{d-k-1}\lambda_i^{(2)\downarrow}\,.
\end{equation}
Due to the second term on the right hand-side we have to distinguish the two cases $d-1\geq N+k$ and $d-1<N+k$ for all $k\in \{1, ..., d-1\}$. Together with the normalizations $\|\bd \lambda^{(1)}\|_1 = \|\bd\lambda^{(2)}\|_1 = N$, we obtain
\begin{equation}
\sum_{i=k+1}^d\lambda_i^\downarrow\leq w\sum_{i=k+1}^d\lambda_i^{(1)\downarrow}+ (1-w)\chi_{N,d}^{(k)}\,,
\end{equation}
where
\begin{equation}
\chi_{N,d}^{(k)} = \begin{cases}
N \quad &\text{if }d-1\geq N+k\,,\\
d-k-1\quad &\text{if }d-1< N+k
\end{cases}\,.
\end{equation}
Therefore, it follows that the sum of the $k$-largest entries of $\bd\lambda$ is bounded from below by
\begin{equation}
\sum_{i=1}^k\lambda_i^\downarrow\geq N - w\left(N-\sum_{i=1}^k\lambda_i^{(1)\downarrow}\right)-(1-w)\chi_{N,d}^{(k)}\,,
\end{equation}
where we used the normalization $\sum_{i=1}^d\lambda_i=N$.

Therefore, the residuals $\mathcal{R}_k(\bd\lambda)$ satisfy the lower bounds
\begin{equation}
\mathcal{R}_k(\bd\lambda) \geq \begin{cases} w\left( k - \sum_{i=1}^k\lambda_i^{(1)\downarrow}\right)\,\,\text{if }k\in \{1, ..., N-1\}\,,\\
w-1+w\left(N - \sum_{i=1}^k\lambda_i^{(1)\downarrow}\right)\,\,\text{if }k=N\,,\\
w\left( N - \sum_{i=1}^k\lambda_i^{(1)\downarrow}\right)\,\,\text{if }k\in \{N+1, ..., d-1\}
\end{cases}\,.
\end{equation}
Moreover, we obtain the following upper bounds
\begin{widetext}
\begin{equation}
\mathcal{R}_k(\bd\lambda) \leq \begin{cases} k-N + w\left(N - \sum_{i=1}^k\lambda_i^{(1)\downarrow}\right) +(1-w) \chi_{N,d}^{(k)}\,\,&\text{if } k\in \{1, ..., N-1\}\,,\\
w\left(N-  \sum_{i=1}^k\lambda_i^{(1)\downarrow}\right)  +(1-w)\left(\chi_{N,d}^{(N)}-1\right)  \,\,&\text{if } k=N\\
w\left(N-  \sum_{i=1}^k\lambda_i^{(1)\downarrow}\right)  +(1-w)\chi_{N,d}^{(k)}\,\,&\text{if } k\in \{N+1, ..., d-1\}
\end{cases}\,.
\end{equation}
\end{widetext}

\section{Derivation of vertex representation of $\Sigma(\bd w, \bd\lambda^{(1)})$ \label{app:v-rep}}

In this section, we derive the vertex representation of $\Sigma(\bd w, \bd\lambda^{(1)})$ for $r=2$ in Sec.~\ref{sec:spectral}.
From the hyperplane representation of the convex set $\Sigma(\bd w, \bd\lambda^{(1)})$ in Eq.~\eqref{eq:Sigma-Xi}, we obtain that $\bd\lambda\in \Sigma(\bd w, \bd\lambda^{(1)})$ if for all $ 1\leq M\leq N-1$,
\begin{equation}
\sum_{i=1}^M\lambda_i^\downarrow\leq \min\left(M,w\sum_{i=1}^M\lambda_i^{(1)\downarrow}+M(1-w)  \right)\,.
\end{equation}
The smaller entry on the right hand-side in the above equation is always given by the second entry in the minimum, i.e., by $w\sum_{i=1}^M\lambda_i^{(1)\downarrow}+M(1-w)$. This implies that the $(N-1)$ largest entries of the generating vertex $\widetilde{\bd v}$ of the permutohedron $\Sigma(\bd w, \bd\lambda^{(1)})$ are given by
\begin{equation}\label{eq:app-vnm1}
\widetilde{v}_i = w\lambda_i^{(1)\downarrow}+(1-w)\quad \forall 1\leq i\leq N-1\,.
\end{equation}
To derive the $N$-th largest entry of $\widetilde{\bd v}$ we use that
\begin{equation}
\sum_{i=1}^N\lambda_i^\downarrow\leq \min\left(N-1+w, w\sum_{i=1}^N\lambda_i^{(1)\downarrow}+N(1-w)   \right)\,,
\end{equation}
where the first entry of the minimum on the right hand-side arises from the hyperplane representation of $\Sigma(\bd w)$. Using Eq.~\eqref{eq:app-vnm1} this implies that
\begin{equation}
\widetilde{v}_N = \min\left( w\lambda_N^{(1)\downarrow}+(1-w), N-1+w -\sum_{j<N}\widetilde{v}_j  \right)\,.
\end{equation}
Moreover, the hyperplane representation of $\Sigma(\bd w, \bd\lambda^{(1)})$ implies that
\begin{equation}
\sum_{i=1}^{N+1}\lambda_i^\downarrow\leq \min\left(N, w\sum_{i=1}^{N+1}\lambda_i^{(1)\downarrow}+N(1-w)   \right)\,.
\end{equation}
Together with $\widetilde{v}_i$ for $i\leq N$, we obtain as an intermediate result that
\begin{align}
\widetilde{v}_{N+1} &= \min\left(N -\sum_{j<N+1}\widetilde{v}_j, \right.\nonumber\\
&\left. \quad \quad w(\lambda_{N+1}^{(1)\downarrow}+\lambda_N^{(1)\downarrow}) +(1-w)-\widetilde{v}_N\right)\,.
\end{align}
Due to the normalization of $\bd\lambda^{(1)}$ to the total particle number, i.e., $\sum_{i=1}^d\lambda_i^{(1)\downarrow}=N$ it holds that
\begin{equation}\label{eq:app-int}
N - \sum_{j<N+1}\widetilde{v}_j\geq w(\lambda_{N+1}^{(1)\downarrow}+\lambda_N^{(1)\downarrow}) +(1-w)-\widetilde{v}_N\,.
\end{equation}
We then define the difference between $w\lambda_N^{(1)\downarrow}+(1-w)$ and $\widetilde{v}_N$ as
\begin{equation}
\widetilde{x}\equiv w\lambda_N^{(1)\downarrow}+(1-w)-\widetilde{v}_N\geq 0\,.
\end{equation}
Together with Eq.~\eqref{eq:app-int} this implies that
\begin{equation}
\widetilde{v}_{N+1} = w\lambda_{N+1}^{(1)\downarrow} + \widetilde{x}\,.
\end{equation}
Moreover, it follows that
\begin{equation}
\widetilde{v}_{i} = w\lambda_i^{(1)\downarrow}\quad \forall N+2\leq i\leq d\,.
\end{equation}

As a consistency check, we verify that the generating vertex $\widetilde{\bd v}$ is decreasingly ordered. The first $N$ entries of $\widetilde{\bd v}$ are automatically decreasingly ordered. The same is true for the vector entries $\widetilde{v}_i\geq \widetilde{v}_j$ for $i<j,  i, j \in \{N+1, ..., d\}$. Therefore, it remains to show that $\widetilde{v}_N\geq \widetilde{v}_{N+1}$. To this end, we first show that
\begin{align}
\widetilde{x} &\leq w\sum_{i=1}^N\lambda_i^{(1)\downarrow}+(1-w)-Nw\nonumber\\
&\leq 1-w\,,
\end{align}
where we used that
\begin{equation}
\widetilde{v}_N\geq wN - w\sum_{j=1}^{N-1}\lambda_j^{(1)\downarrow}\,.
\end{equation}
Using the upper bound on $\widetilde{x}$ we finally arrive at
\begin{align}
\widetilde{v}_{N+1} - \widetilde{v}_{N}&\leq w\left(\lambda_{N+1}^{(1)\downarrow}-\lambda_N^{(1)\downarrow}\right) +\widetilde{x}-(1-w)\nonumber\\
&\leq 0\,.
\end{align}
Thus, $\widetilde{\bd v}$ is indeed decreasingly ordered as required. This concludes the derivation of the vertex representation of $\Sigma(\bd w, \bd\lambda^{(1)})$ for $r=2$.

\bibliography{Refs}

\end{document}